\documentclass[10pt,rqno, a4paper]{article}

\usepackage{xr}
\usepackage{titlesec}
\titleformat*{\section}{\bf\large\centering} 
\titleformat*{\subsection}{\bf} 
\titleformat*{\subsubsection}{\it} 

\usepackage{psfrag}
\usepackage{cite}
\usepackage[small]{caption}
\usepackage{amsmath}
\usepackage{amssymb}
\usepackage{amsfonts}
\usepackage{setspace}
\usepackage{latexsym}
\usepackage{mathrsfs}
\usepackage{epsfig}
\usepackage{amsthm}
\usepackage{yfonts}
\usepackage{graphicx}
\usepackage{color}
\usepackage{verbatim}
\usepackage{enumitem}
\usepackage[margin=2cm]{caption}

\newcommand{\be}{\begin{equation}}
\newcommand{\ee}{\end{equation}}

\newcommand{\vs}{\vspace{0.2cm}}

\usepackage{pstool}
\usepackage{fancyhdr}
\AtBeginDocument{\thispagestyle{plain}}
\fancypagestyle{plain}
	{\fancyhead{}
	\fancyfoot{}
	\fancyhead[LE,RO]{}
	\fancyhead[LO,RE]{}
	\fancyfoot[C]{\thepage}
	
	\headsep = 20pt}

\numberwithin{equation}{subsection}
\newtheorem{Theorem}{Theorem}[subsection]
\newtheorem{Definition}[Theorem]{Definition}
\newtheorem{Proposition}[Theorem]{Proposition}

\newtheorem{Lemma}[Theorem]{Lemma}
\newtheorem{Corollary}[Theorem]{Corollary}
\newtheorem{Problem}[Theorem]{Problem}

\newcommand{\dist}{d} 

\newcommand{\sg}{g} 
\newcommand{\hg}{{\mathfrak{g}}} 

\newcommand{\length}{L} 
\newcommand{\diam}{{\rm  diam}} 

\newcommand{\T}{{\rm T}}
\newcommand{\Sa}{{\rm S}^{1}} 

\newcommand{\sM}{\Sigma} 

\newcommand{\gcur}{\kappa} 

\usepackage{setspace, tocloft}

\usepackage{tocloft} 

\allowdisplaybreaks

\begin{document}

\thispagestyle{empty}

\begin{center}
{\Large\bf A classification theorem for static vacuum black holes
\vspace{.2cm}

Part I: the study of the lapse
}

\vspace{.5cm}

{\sc Mart\'in Reiris Ithurralde}

{mreiris@cmat.edu.uy}
\vs

{\it Centro de Matem\'atica/Universidad de la Rep\'ublica}

{\it Montevideo, Uruguay}
\vs

\begin{abstract}
The celebrated uniqueness's theorem of the Schwarzschild solution by Israel, Robinson et al, and Bunting/Masood-ul-Alam, asserts that the only asymptotically flat static solution of the vacuum Einstein equations with compact but non-necessarily connected horizon is Schwarzschild. Between this article and its sequel we extend this result by proving a  classification theorem for all (metrically complete) solutions of the static vacuum Einstein equations with compact but non-necessarily connected horizon without making any further assumption on the topology or the asymptotic. It is shown that any such solution is either: (i) a Boost, (ii) a Schwarzschild black hole, or (iii) is of Myers/Korotkin-Nicolai type, that is, it has the same topology and Kasner asymptotic as the Myers/Korotkin-Nicolai black holes. In a broad sense, the theorem classifies all the static vacuum black holes in 3+1-dimensions. 

In this Part I we use introduce techniques in conformal geometry and comparison geometry \'a la Bakry-\'Emery to prove, among other things, that vacuum static black holes have only one end, and, furthermore, that the lapse is bounded away from zero at infinity. The techniques have interest in themselves and could be applied in other contexts as well, for instance to study higher-dimensional static black holes.
\end{abstract}

\end{center}

\newpage

\begin{center}
\begin{minipage}[l]{11cm}
{\small\tableofcontents}
\end{minipage}
\end{center}
\vs

\section{Introduction}  

The vacuum static solutions of the Einstein equations have played since early days a fundamental role in the study of Einstein's theory and the classification theorems have been at the center of the work. In this context, the celebrated {\it uniqueness theorem of the Schwarzschild solution} asserts that the Schwarzschild black holes are the only asymptotically flat vacuum static solutions with compact but non-necessarily connected horizon (Israel \cite{Israel}, Robinson et al \cite{RobinsonII}, Bunting/Masood-ul-Alam \cite{MR876598}; for a review on the history of this theorem see \cite{Robinson}). Between this article and its sequel it is proved a classification theorem extending Schwarzschild's uniqueness theorem to vacuum static solutions having compact but non-necessarily connected horizon without making further assumptions on their topology or asymptotic. 

Static solutions appear in many contexts. In Riemannian geometry they model for instance the blow up of singularities forming along sequences of Yamabe metrics \cite{MR1452867},\cite{MR1726233}, \cite{MR1837365}, and provide interesting examples of Ricci-flat Riemannian metrics with a warped $\Sa$-factor \cite{MR1809792}. In physics they are crucial for example in the study of mass, quasi-local mass and initial data sets \cite{MR996396}, \cite{MR3064190}, \cite{MR3037574}, or in the exploration of certain high-dimensional theories \cite{PhysRevD.35.455}. A classification theorem can be relevant in any of these contexts.

Stated below is the classification theorem that we shall prove. The objects to classify are {\it static black hole data sets} that condensate the notion of static black hole at the initial data level\footnote{which is the viewpoint adopted in these articles. We classify static black hole spacetimes having a Cauchy hypersurface orthogonal to the static Killing field. The problem of classifying static spacetimes without such condition is not treated here, see for instance \cite{MR3077927}.}. Their definition and the discussion of the three main families in the theorem is given right after. Full technical details can be found in the background subsection \ref{SDSMT}. Previous work and references related to these articles are discussed at the end of this section. For better clarity the proof's structure of the classification theorem is explained separately in the next subsection \ref{TSOTP}. A detailed account of the contents of this Part I is given in subsection \ref{CSTA}.
\begin{Theorem}[The classification Theorem]\label{TCTHM} Any static black hole data set is either,
\begin{enumerate}[labelindent=\parindent, leftmargin=*, label={\rm (\Roman*)}, widest=a, align=left]
\item\label{FTII} A Schwarzschild black hole, or,
\item\label{FTI} a Boost, or,
\item\label{FTIII} is of Myers/Korotkin-Nicolai type.
\end{enumerate}
\end{Theorem}

Formally, a {\it (vacuum) static data set} $(\Sigma; g, N)$ consists of an orientable three manifold $\Sigma$, a function $N$ called the lapse and positive in the interior $\Sigma^{\circ}=\Sigma\setminus \partial \Sigma$ of $\Sigma$, and a Riemannian metric $g$ on $\Sigma$ satisfying the vacuum static equations,
\begin{gather}
\label{EEII} NRic = \nabla\nabla N,\quad \Delta N=0
\end{gather}
A static data set $(\Sigma;g,N)$ gives rise to a vacuum static spacetime (${\bf Ric}=0$), 
\be\label{SPTE}
{\bf \Sigma}=\mathbb{R}\times \Sigma,\quad {\bf g}=N^{2}dt^{2}+g,
\ee
where $\partial_{t}$ is the static Killing field. Conversely, a static spacetime of the form (\ref{SPTE}), gives rise to a static data set $(\Sigma;g,N)$. Throughout this article we will work with static data sets rather than their associated spacetimes. 

A {\it static black hole data set} is defined as a static data $(\Sigma;g,N)$ such that $\partial \Sigma=\{N=0\}\neq \emptyset$ is compact and $(\Sigma;g)$ is metrically complete. In this definition no special asymptotic or global topological structure is assumed. The boundary of $\Sigma$ is non-necessarily connected and is called the horizon. Without further justification, we will say that the spacetime of a static black hole data set is a `black hole spacetime', \footnote{Indeed the outer-communication region.}. We stress that all the analysis in these articles is carried only on static data sets, leaving the spacetime picture aside. 

Let us discuss now the families \ref{FTII}, \ref{FTI} and \ref{FTIII} of static black hole data sets. 

The Schwarzschild static black hole data sets are spherically symmetric and asymptotically flat, and are given explicitly by,
\be\label{SCHDT}
\Sigma=\mathbb{R}^{3}\setminus B(0,2m),\quad g=\frac{1}{1-2m/r}dr^{2}+r^{2}d\Omega^{2}\quad {\rm and}\quad N=\sqrt{1-2m/r}
\ee
where $m>0$ is the mass and $B(0,2m)$ is the open ball of radius $2m$\footnote{The spacetime (\ref{SPTE}) corresponding to (\ref{SCHDT}) is just the region of exterior communication of a Schwarzschild black hole of mass $m$. The horizon is the boundary $\partial \Sigma = \{N=0\}$. Restricted to $r\geq R(t)> 2m$, the Schwarzschild space models the gravitational field of any isolated but spherically symmetric physical body of radius $R(t)$. The object itself may be transiting a dynamical process (for instance in a star), but the spacetime outside remains spherically symmetric and thus Schwarzschild by Birkhoff's theorem. If the radius $R(t)$ goes below the threshold of $2m$, no equilibrium is possible, the body undergoes a complete gravitational collapse and a Schwarzschild black hole remains.}. The family is parameterised by the mass $m>0$. It is of course the paradigmatic family of static black hole data sets.
\begin{figure}[h]
\centering
\includegraphics[width=3.7cm,height=3.7cm]{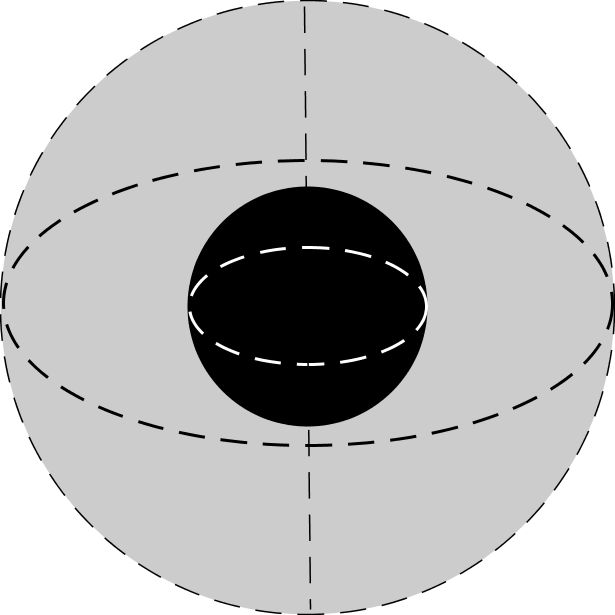}
\caption{A Schwarzschild black hole. The grey region is $\Sigma$ and is diffeomorphic to $\mathbb{R}^{3}$ minus the open (black) ball $B(0,2m)$. The solution is spherically symmetric and thus axisymmetric.}
\label{FigureBoost}
\end{figure}

The flat static data
\be\label{BOOSTDEF} 
\Sigma=[0,\infty)\times \mathbb{R}^{2};\quad g= dx^{2}+dy^{2}+dz^{2},\quad N=x,
\ee
is called the {\it Boost}. The spacetime (\ref{SPTE}) associated to (\ref{BOOSTDEF}) is the Rindle-wedge of the Minkowski spacetime and the static Killing field is the boost generator $x\partial_{t}$, hence the name. The quotients of the Boost by any $\mathbb{Z}^{2}$ group of isometries generated by two translations along the factor $\mathbb{R}^{2}$, are data of the form,
\be\label{BOOSTQUOTIENT}
\Sigma=[0,\infty)\times \T^{2},\quad g=dx^{2}+h,\quad N=x
\ee
where $h$ is a flat metric on the two-torus $\T^{2}=\Sa\times \Sa$. As the lapse $N$ is zero on the boundary of $\Sigma$, these are static black hole data sets. They define the Boost family in the classification theorem, and is parametrised by the set of flat two-tori.

Other relevant examples of static data sets are the {\it Kasner data sets} (a complete discussion is given in subsection \ref{SSKK} of Part II),
\be\label{Kasner}
\Sigma=(0,\infty)\times \mathbb{R}^{2};\quad g= dx^{2}+x^{2\alpha}dy^{2}+x^{2\beta}dz^{2},\quad N=x^{\gamma},
\ee
where $y$ and $z$ are coordinates on each of the factors $\mathbb{R}$ of $\mathbb{R}^{2}$, and $\alpha, \beta$ and $\gamma$ are any numbers satisfying,
\be
\alpha+\beta+\gamma=1,\qquad \alpha^{2}+\beta^{2}+\gamma^{2}=1
\ee
\begin{figure}[h]
\centering
\includegraphics[width=6cm, height=6cm]{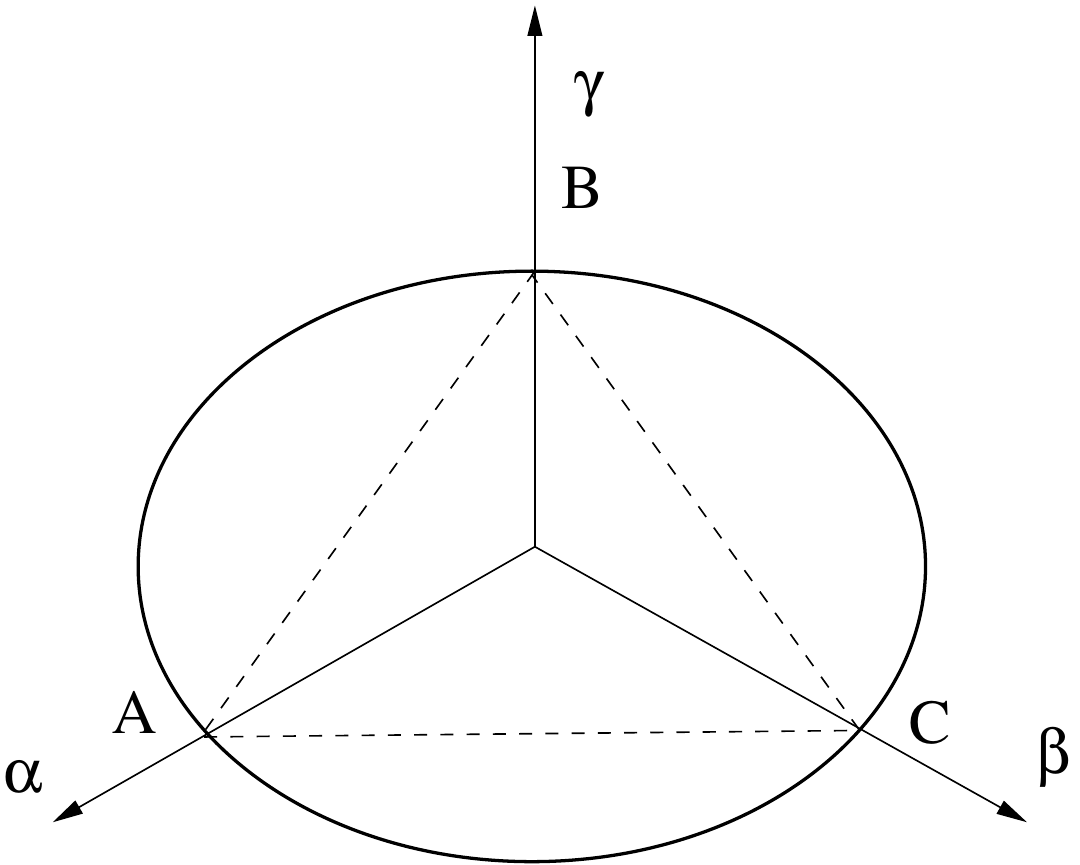}
\caption{The circle that defines the range of the Kasner parameters $\alpha$, $\beta$, $\gamma$.}
\label{Figure21}
\end{figure}
(see Figure \ref{Figure21}). The Kasner space $(\alpha,\beta,\gamma)=(0,0,1)$ is the Boost\footnote{One must add indeed the set $\{0\}\times \mathbb{R}^{2}$.} and is the Kasner data with faster growth of the lapse (linear). We denote it by the letter $B$. The Kasner spaces $(1,0,0)$ and $(0,1,0)$, that have constant lapse and are therefore flat, are denoted respectively by the letters $A$ and $C$. 

As with the Boost, one can quotient a general Kasner data to obtain data of the form,
\be\label{KASNERQUOTIENT}
\Sigma=(0,\infty)\times \T^{2},\quad g=dx^{2}+h(x),\quad N=x^{\gamma}
\ee
where, $h(x)$ is a certain path of flat metrics on $\T^{2}$. This is the Kasner family and is parametrised by the set of possible Kasner triples $(\alpha,\beta,\gamma)$ (a circle) times the set of flat two-tori up to isometry. The Myers/Korotkin-Nicolai data sets, that we describe a few lines below, are asymptotic to them. Finally, we denote also by $A$, $B$, $C$, to the quotients of the spaces $A$, $B$, $C$ respectively.
\begin{figure}[h]
\centering
\includegraphics[width=3.6cm,height=3.6cm]{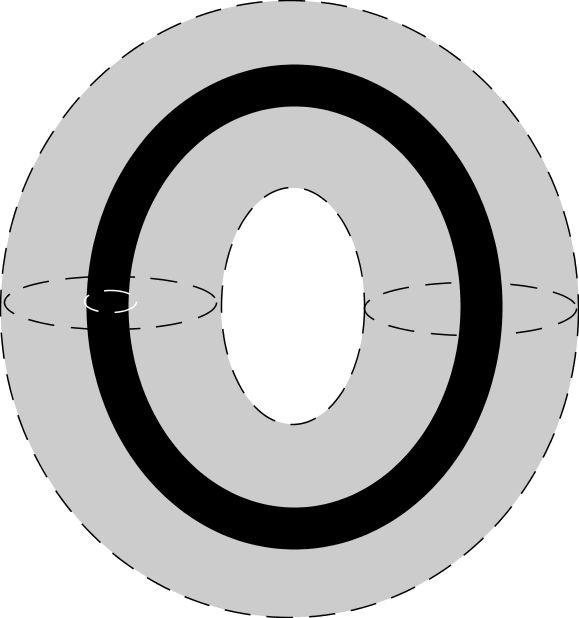}
\caption{A Boost black hole. The grey region is $\Sigma$ and is diffeomorphic to a solid torus minus an open (black) solid torus.}
\label{FigureBoost}
\end{figure}       

Let us see the last family in the classification theorem, namely the static black hole data sets of Myers/Korotkin-Nicolai type. A static black hole data set is said to be of {\it Myers/Korotkin-Nicolai type} if its topology is that of a solid three-torus minus a finite number of balls and is asymptotic to a Kasner space (\ref{KASNERQUOTIENT}), (see Definition \ref{KADEF}).  Black holes with such properties were found by Myers in \cite{PhysRevD.35.455} and were rediscovered and further investigated by Korotkin and Nicolai in \cite{94aperiodic}, \cite{KOROTKIN1994229}. Myers and Korotkin/Nicolai's construction used first Weyl's method to find a `periodic' static solution by superposing along a common axis an infinite number of Schwarzschild solutions separated by the same distance $L$ (see Figure \ref{UMKN}). Simple quotients give then the desired solutions with any number of holes (see Figure \ref{Myers-Korotkin-Nicolai}),
\footnote{As the Schwarzschild solutions are axisymmetric, they can be superposed along an axis by Weyl's method. When superposing a finite number of holes, angle deficiencies appear on the axis between them and the solution resulting is non-smooth. This deficiency can be understood from the fact that a repulsive force must keep the holes in equilibrium. However when infinitely many of them are superposed along the axis, say at a distance $L$ from each other, no extra force is needed and the angle deficiency is no longer present. This gives a `periodic' solution that can be quotient to obtain M/KN solutions with any number of holes.}. 

The details of such data sets $(\Sigma; g ,N)$ are mainly irrelevant to us but for the sake of completeness the main features of the data in the universal cover space can be summarised as follows (see \cite{PhysRevD.35.455},\cite{94aperiodic}).
\begin{figure}[h]
\centering
\includegraphics[width=2cm,height=7cm]{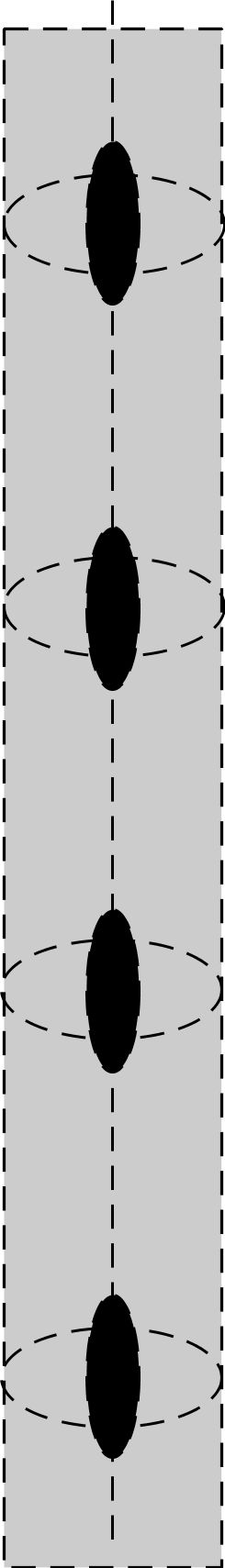}
\caption{A 'universal M/KN data'. The grey region is $\Sigma$ and is diffeomorphic to $\mathbb{R}^{3}$ minus an infinite number of (black) open balls. The solution is axisymmetric.}
\label{UMKN}
\end{figure} 
\begin{figure}[h]
\centering
\includegraphics[width=3.6cm,height=3.6cm]{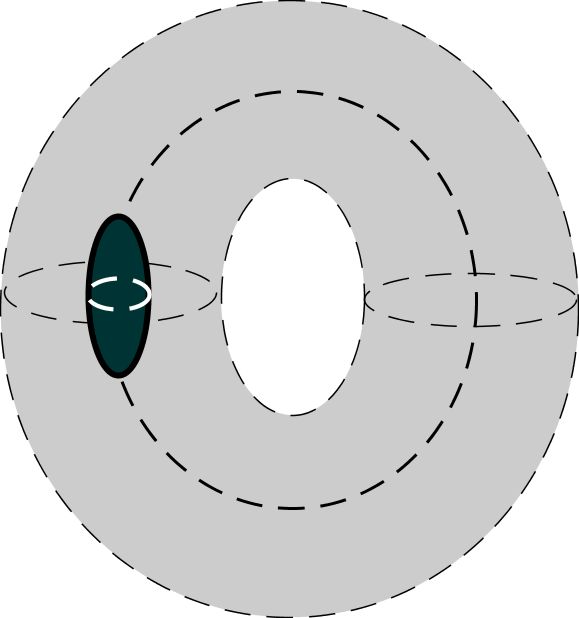}
\caption{A M/KN data with one hole. The grey region is $\Sigma$ and is diffeomorphic to a solid torus minus an open (black) ball. The solution is axisymmetric.}
\label{Myers-Korotkin-Nicolai}
\end{figure} 
The metric and the lapse have the form,
\be
g =e^{-\omega}(e^{2k}(dx^{2}+d\rho^{2})+\rho^{2}d\phi^{2}),\qquad N=e^{\omega/2},
\ee
where $(x,\rho)$ are Weyl coordinates ($\rho>0$ is the radial coordinate) and $\phi\in [0,2\pi)$ is the angular coordinate. The function $\omega$ is defined through the convergent series,
\be
\omega(x,\rho)=\omega_{0}(x,\rho)+\sum_{n=1}^{\infty}[\omega_{0}(x+nL,\rho)+\omega_{0}(x-nL,\rho)+\frac{4M}{nL}]
\ee
where $\omega_{0}(x,\rho)$ is,
\be
\omega_{0}=\ln \mathcal{E}_{0},\qquad \mathcal{E}_{0}(x,\rho)=\frac{\sqrt{(x-M)^{2}+\rho^{2}}+\sqrt{(x+M)^{2}+\rho^{2}}-2M}{\sqrt{(x-M)^{2}+\rho^{2}}+\sqrt{(x+M)^{2}+\rho^{2}}+2M}
\ee
and the function $k(x,\rho)$ is found by quadratures through the equations,
\be
k_{\rho}=\frac{\rho}{4}(\omega_{\rho}^{2}-\omega_{x}^{2}),\qquad k_{x}=\frac{\rho}{2}\omega_{x}\omega_{\rho},
\ee
The metric $g$, the lapse $N$ and the function $k$ are invariant under the translations $x\rightarrow x+L$, hence periodic. The asymptotic of the solution is Kasner and has the form,
\be
g \approx c_{1}\rho^{\alpha^{2}/2-\alpha}(dx^{2}+d\rho^{2})+c_{2}\rho^{2-\alpha}d\phi^{2},\qquad N \approx c_{3}\rho^{\alpha/2}
\ee
where $\alpha=4M/L$ and so $0< \alpha<2$. Note that the range of $\alpha$ excludes the Kasner spaces $A$, $B$ and $C$, and clearly those with $\gamma<0$ for which $N\rightarrow 0$ at infinity. Therefore the asymptotic of such static black hole data sets is Kasner but different from $A$, $B$, $C$ and those Kasner with $\gamma<0$. This fact was not incorporated in the definition of static black hole data set of M/KN type. It will be shown however in Part II that the Kasner asymptotic of a black hole of M/KN type is indeed different from $A$ and $C$, although we cannot exclude the possibility of being asymptotic to $B$. Of course by the maximum principle, the Kasner asymptotic cannot be one with $\gamma<0$ (if so then it must be $N=0$ on $\Sigma$ because $N=0$ on $\partial \Sigma$ and $N\rightarrow 0$ at infinity). We leave it as an open problem to prove that the only static black hole data sets asymptotic to a Boost are in fact the Boosts. 

The construction of Myers/Korotkin-Nicolai that we briefly described above can be generalised to allow a periodic superposition of Schwarzschild holes of different masses provided they are kept separated from each other at the right distances. The outcome, (after quotient), are static black hole data sets of M/KN type different from the ones just described. To embrace all the possibilities we define the {\it Myers/Korotkin-Nicolai data sets} as any axisymmetric static black hole data set obtained using Myers/Korotkin-Nicolai's method. It could be that such data sets are the only black hole static data sets of M/KN type. We leave this as an open problem (see Problem \ref{OPENPRO}). Note that the precise global geometry of the M/KN data sets won't be discussed in this article and won't play a role (for a discussion see \cite{KOROTKIN1994229}) as we will deal only with data sets of M/KN-type that are defined by abstracting the main geometric features of the M/KN data sets.  
\vs

The proof of the classification theorem is divided between Part I (this article) and Part II (its sequel), and each article has a clear and distinct motivation. The main purpose of this Part I, that we elaborate in detail in the subsections \ref{TSOTP} and \ref{CSTA} below, is to study global properties of  the lapse of static black hole data sets and its implications on the global geometry. Part II discusses, on one side, $\Sa$-symmetric static data sets and, on the other side, provides a detailed study of the asymptotic of static ends. Part I uses techniques in conformal geometry and comparison geometry \'a la Bakry \'Emery, whereas Part II uses techniques in standard comparison geometry and convergence and collapse of Riemannian manifolds. Several sections inside each part are new and have their own interest going behind the main purpose of these articles. To make it more clear, the proof's structure of the classification theorem is explained separately in subsection \ref{TSOTP} below.
\vs

These articles continue in a sense our work on static solutions in \cite{MR2919527}, \cite{Reiris2017}, \cite{MR3233266}, and \cite{MR3233267}. In particular, in \cite{MR3233266} and \cite{MR3233267} it was shown that asymptotic flatness in Schwarzschild's uniqueness theorem can be replaced (still preserving uniqueness) by the metric completeness of $(\Sigma;g)$ plus the condition that, outside a compact set, $\Sigma$ is diffeomorphic to $\mathbb{R}^{3}$ minus a ball. Without any topological hypothesis Schwarzschild's uniqueness of course fails. Thus \cite{MR3233266} and \cite{MR3233267} prove a classification theorem somehow in between Schwarzschild's uniqueness theorem and the classification Theorem \ref{TCTHM}. We do not know of any attempt in the literature pointing to a general classification theorem of static vacuum black holes, except, perhaps, a conjecture stated by Anderson in \cite{MR1452867} (Conjecture 6.2), that appears to be incomplete. Still, vacuum static solutions have been deeply investigated along the years, so to conclude this introduction let us recall former developments that are related technically or conceptually to this work. We point out connections when it is appropriate. 

Vacuum static solutions with symmetries have been investigated since early days by Schwarzschild \cite{Schwarzschild}, Levi-Civita \cite{Levi-Civita2011a}, \cite{Levi-Civita2011b}, Kasner \cite{MR1501305}, \cite{MR1501301}, Weyl \cite{ANDP:ANDP19173591804} and many others, and there is an advanced understanding of them (for a review see \cite{Jordan2009} and references therein). Understanding static solutions without any a priori symmetry is vast more complex. Schwarzschild's uniqueness theorem was perhaps the first general classification theorem although it demands global assumptions. Israel's seminal work required that the lapse $N$ can be chosen as a global coordinate and therefore required a connected spherical horizon. This technical global condition on the lapse was removed later by M\"uller, Robinson and Seifert in \cite{Hagen1973}, but keeping the hypothesis of a connected horizon. A simpler proof of their result was found later by Robinson by means of a remarkable integral formula \cite{RobinsonII} (the proof used also previous work by K\"unzle \cite{kunzle1971}). Altogether, this proved that the only asymptotically flat solution with a connected compact horizon is Schwarzschild. The analysis of the geometry of the level sets of the lapse function that play a fundamental role in \cite{Israel} and \cite{RobinsonII} and in other works on static solutions as well, will be also relevant here when we study Kasner asymptotic in subsection \ref{ENDSAK} of Part II. We will follow however different techniques. Other proofs of the Israel-Robinson theorem were given more recently by the author in \cite{MR2919527} and by Agostiniani and Mazzieri in \cite{2015arXiv150404563A}. In \cite{MR2919527} techniques in comparison geometry were used and in \cite{2015arXiv150404563A} monotonic quantities along the level sets of the lapse were introduced. Some of the arguments in this article will follow similar ideas though technically distinct. The uniqueness of Schwarzschild even when multiple horizons are in principle allowed was settled by Bunting/Masood-ul-Alam\cite{MR876598}, using the positive mass theorem. 

As mentioned earlier, there seems to be no previous attempt in the literature to classify static black holes data sets that are not asymptotically flat, except perhaps, the conjecture in \cite{MR1452867}. Connected to that work, Anderson performed a general study of static and stationary solutions in \cite{MR1809792} and \cite{MR1806984} respectively, obtaining a fundamental decay estimate for the curvature and the gradient of the logarithm of the lapse. Among other things, this establishes the first uniqueness theorem of the Minkowski solution (as a static solution) without assuming any type of asymptotic but just geodesic completeness. In \cite{Reiris2017} it was shown that Anderson's estimate holds too in any dimension by importing techniques in comparison geometry \'a la Backry-\'Emery that were introduced by J. Case in \cite{MR2741248} in a context somehow related to that of static solutions. These new techniques in comparison geometry a la Bakry-Emery play a fundamental role in this Part I as we will explain below. The global study of the lapse function that we do is based largely upon these ideas.

\subsection{The proof's structure of the classification theorem}\label{TSOTP}

The proof of the classification theorem is divided in three steps. Say $(\Sigma;g,N)$ is a static black hole data set. Then the proof requires proving that,
\begin{enumerate}
\item\label{Step-a} $\Sigma$ has only one end.
\item\label{Step-b} The horizons are {\it weakly outermost} (see Definition \ref{DWO}).
\item\label{Step-c} The end is asymptotically flat or asymptotically Kasner. 
\end{enumerate}
Once this is achieved the proof of the classification theorem is direct from known results. Indeed, assume \ref{Step-a}-\ref{Step-c} hold. If the data is asymptotically flat, it follows that it must be Schwarzschild by the uniqueness theorem. If the data is asymptotically Kasner, then it is deduced that it is either a Boost or is of M/KN type as follows. First, by step \ref{Step-b} the horizons are weakly outermost, and thus by Schoen-Galloway \cite{MR2238889} and Galloway \cite{4b6cb19bc94d4cf485e58571e3062f77}, either the data is a Boost or every horizon is a totally geodesic sphere. Let us assume the data is not a Boost. If the Kasner asymptotic is different from $B$, then, as any constant $x$-coordinate torus of any Kasner space different from $B$ has positive outwards mean curvature (from (\ref{Kasner}) the mean curvature is $\theta=(\alpha+\beta)/x$ with $\alpha+\beta>0$ if $(\alpha,\beta,\gamma)\neq (0,0,1)$), we can clearly find (using the fast decay into the Kasner space) a two-torus $T$ separating $\Sigma$ into two manifolds, $\Sigma_{1}$ and $\Sigma_{2}$, with $\overline{\Sigma}_{2}$ diffeomorphic to $[0,\infty)\times T$ and $\overline{\Sigma}_{1}$ a compact manifold whose boundary consist of $T$, of positive outwards mean curvature, and a finite number of spherical-horizons. It then follows from Galloway's \cite{MR1201655} that $\overline{\Sigma}_{1}$ is diffeomorphic to a solid three-torus minus a finite number of open three-balls\footnote{Galloway's results precisely asserts that if a static data set $(\Sigma;N,g)$ is such that $\Sigma$ is compact and $\partial \Sigma$ consists of a convex sphere plus $h$ horizons, then $\Sigma$ is diffeomorphic to a closed three-ball minus $h$-open three-balls. If instead of having a convex spherical component of $\partial \Sigma$ there is a convex toroidal component, the one can use Galloway's argumentation (without any substantial change) to show that $\Sigma$ is diffeomorphic to a closed solid three-torus minus a finite number of open three-balls.\label{FN}}. Hence, $\Sigma$ is diffeomorphic to an open three-torus minus a finite number of open three-balls. This type of topology and the Kasner asymptotic imply, by definition, that the data is of M/KN type. If the Kasner asymptotic is $B$, then there are no obvious embedded tori $T$ of positive outwards mean curvature, but it will be proved that there are in fact tori $T$ separating $\Sigma$ in $\Sigma_{1}$ and $\Sigma_{2}$ as before, but having area strictly less than the asymptotic area of the `transversal' tori over the end. This is enough to repeat Galloway's argument and conclude that indeed $\Sigma$ has the desired topology. 

The main motivation of this article (Part I) is to prove the steps \ref{Step-a}, \ref{Step-b}. We do that in section \ref{CTPL}. The proof of step \ref{Step-c} is done in section \ref{VWAE} of Part II and requires using section \ref{S1S} of Part II at some particular instances. Part II uses Part I as follows. Until subsection \ref{FTKASS}, it is either not used, or it is used only that if $\partial \Sigma$ is compact, then the metric $\hg$ is complete at infinity. This is shown in Theorem \ref{COMN2} of subsection \ref{CMMC} of Part I. Subsection \ref{POKA}, proving the Kasner asymptotic of static black hole ends with sub-cubic volume growth, uses the completeness of $\hg$ at infinity, and steps \ref{Step-a} and \ref{Step-b}.

We pass now to discuss the structure of the different sections of this article and the main points behind the various proofs. 

\subsection{The contents and the structure of this article (Part I)}\label{CSTA}

Section \ref{BACKGROUNDMATERIAL} contains the background material, including notation and terminology. Subsection \ref{SDSMT} contains the main definitions, as the one of static black hole data set or Kasner asymptotic, and states again the classification theorem as Theorem \ref{TCTHM}. Subsection \ref{SAP} defines annuli and partitions cuts, that are useful to study asymptotic properties.
 
The body of the article begins in section \ref{CTPL} where we discuss the properties of metrics $\overline{g}$ conformally related to a static metric $g$ by powers of the lapse, namely $\overline{g}=N^{-2\epsilon}g$ where $\epsilon$ is just a constant. The reasons why we study these conformal metrics are mainly the following. First, we will use the metrics $\overline{g}=N^{-2\epsilon}g$ with $\epsilon>0$ to accomplish step \ref{Step-a} (of subsection \ref{TSOTP}), that is, proving that static black hole data sets have only one end, Theorem \ref{KUNO}. Second, the proof of step \ref{Step-b}, that the horizons of black hole data sets are weakly outermost, requires proving in particular the metric completeness of $\hg=N^{2}g$ (i.e. $\epsilon=-1$) away from the boundary\footnote{Namely $(\Sigma_{\delta};\overline{g})$ is metrically complete where $\Sigma_{\delta}$ is $\Sigma$ with a collar around the boundary removed. Note that the metric $\hg$ is singular at $\partial \Sigma$, so to speak about completeness we need to remove a collar around $\partial \Sigma$.}. This is done in Proposition \ref{SOFOR} again using the metrics $\overline{g}=N^{-2\epsilon}g$ with $\epsilon$ in a certain range, Theorem \ref{COMN2}. Third, in section \ref{VWAE} of Part II, and because of its nice properties, we will use mainly $\hg$ to study the asymptotic of black hole data sets. Once more, it is necessary to grant that $\hg$ is complete at infinity.

The results of Section \ref{CTPL}, in particular the investigation of the conformal metrics $\overline{g}$, rely in casting the static equations in a framework \'a la Bakry-\'Emery, and then using some general properties of these spaces in a suitable way. Let us make this more precise. Using $f=-\ln N$ instead of the variable $N$, the static equations read,
\be\label{EESTR}
Ric^{1}_{f}=0,\qquad \Delta_{f}f=0
\ee
where for any $\alpha$  the $\alpha$-Bakry-\'Emery Ricci tensor $Ric^{\alpha}_{f}$ is,
\be
Ric^{\alpha}_{f}:=Ric+\nabla\nabla f-\alpha\nabla f \nabla f,\\
\ee
whereas the $f$-Laplacian $\Delta_{f} \phi$ of a function $\phi$ is,
\be
\Delta_{f}\phi:=\Delta\phi-\langle \nabla f,\nabla \phi\rangle
\ee
If instead of $g$ and $f=-\ln N$ we use the variables $\overline{g}=N^{-2\epsilon}g$ and $f=-(1+\epsilon)\ln N$, then the static equations are,
\be\label{EESTR2}
\overline{Ric}^{\alpha}_{f}=0,\qquad \overline{\Delta}_{f} f=0
\ee
where $\alpha=(1-2\epsilon-\epsilon^{2})/(1+\epsilon)^{2}$. The constant $\alpha$ is positive for $\epsilon$ in the range $-1-\sqrt{2}<\epsilon<-1+\sqrt{2}$. The equations (\ref{EESTR}) and (\ref{EESTR2}) share the same structure (only the $\alpha$ is different), and is the right way to present these equations to apply techniques \'a la Bakry-\'Emery. Spaces having $Ric^{\alpha}_{f}\geq 0$ with $\alpha>0$, have been studied in recent years under the context of comparison geometry (see \cite{MR2577473} and references therein). The crucial fact is that several well known results that hold for spaces with $Ric\geq 0$ hold too for spaces with $Ric^{\alpha}_{f}\geq 0$, $\alpha>0$, no matter the form of $f$. Thus, one can obtain geometric information without assuming any a priori knowledge on $N$. In turn, that information is then used to prove properties of $N$. 

The detailed contents of Section \ref{CTPL} are as follows. Subsection \ref{BESEC} explains the structure of the conformal equations, Proposition \ref{FELIZ}. Subsection \ref{CMACD} proves the crucial Lemma \ref{LEMMAME} (essentially due to Case) and from it it is obtained a generalised Anderson's decay estimate for the conformally related data, Lemma \ref{CDLEMMA}. These estimates are used in subsection \ref{CMMC} to show the metric completeness of the manifolds $(\Sigma; \overline{g}=N^{-2\epsilon}g)$ for $-1-\sqrt{2}<\epsilon<-1+\sqrt{2}$ (provided $\partial \Sigma$ is compact, $N|_{\Sigma}>0$ and $(\Sigma;g)$  is metrically complete), Theorem \ref{COMN2}. Until here the results are on general non-necessarily black hole data sets. Subsection \ref{APP} contains important applications to particular situations. First, in subsection \ref{CDPL} remarks are pointed out on the conformal data $(\Sigma; N^{-2\epsilon}g)$ of the data $(\Sigma;g)$ of a static black hole data set, Proposition \ref{PIV}. It is particularly stressed here that, when $\epsilon>0$ is small, the manifold $(\Sigma; N^{-2\epsilon}g)$ is still metrically complete, while the boundary becomes strictly convex (indeed the boundary of $\Sigma$ minus a small collar around $\partial \Sigma$). Then, in subsection \ref{STR} it is proved using the previous subsection and a generalised splitting theorem \'a la Backry-\'Emery that static black hole data sets have only one end, Proposition \ref{KUNO}. This accomplishes step \ref{Step-a}. In subsection \ref{HTT} it is proved using the completeness at infinity of $(\Sigma; N^{2}g=\hg)$ that either black hole data sets are boosts, or every horizon component is a sphere and weakly outermost. This accomplishes step \ref{Step-b}. 
Finally in subsection \ref{TAIS} it is proved that static isolated systems in GR are asymptotically flat. This application is independent of the rest of the article.

Section \ref{GLOBP} proves that the lapse on static black hole data sets is bounded away from zero at infinity. This result is not used per-se in the proof of the classification theorem, although it provides an alternative proof that the metric $\hg$ on static black hole ends is complete at infinity. The section \ref{GLOBP} relies on techniques introduced in the previous section \ref{CTPL} and in a sense can be seen as another application. It could be useful and interesting in other contexts as well, for instance to investigate higher dimensional black hole data sets.
\vspace{.3cm}

{\bf Acknowledgment} I would like to thank Herman Nicolai, Marc Mars, Marcus Kuhri, Gilbert Weinstein, Michael Anderson, Greg Galloway, Miguel Sanchez, Carla Cederbaum, Lorenzo Mazzieri, Virginia Agostiniani and John Hicks for discussions and support. Also my gratefulness to Carla Cederbaum for inviting me to the conference `Static Solutions of the Einstein Equations' (T\"ubingen, 2016), to Piotr Chrusciel for inviting me to the meeting in `Geometry and Relativity' (Vienna, 2017) and to Helmut Friedrich for the very kind invitation to visit the Albert Einstein Institute (Max Planck Institute, Potsdam, 2017). This work has been largely discussed at them. Finally my gratefulness to the support received from the Mathethamical Center at the Universidad de la Rep\'ublica, Uruguay.

\section{Background material}\label{BACKGROUNDMATERIAL}

\subsection{Static data sets and the main Theorem}\label{SDSMT}

Manifolds will always be smooth ($C^{\infty}$). Riemannian metrics as well as tensors will also be smooth.  If $g$ is a Riemannian metric on a manifold $\Sigma$, then 
\be
\dist_{g}(p,q)= \inf\big\{\length_{g}(\gamma_{pq}):\gamma_{pq}\ \text{smooth curve joining $p$ to $q$}\big\},
\ee
is a metric, where $L_{g}$ is the notation we will use for length (when it is clear from the context we will remove the sub-index $g$ and write simply $\dist$ and $L$). A Riemannian manifold $(\Sigma;g)$ is {\it metrically complete} if the metric space $(\Sigma; \dist)$ is complete.

\begin{Definition}[Static data set]\label{SDS} A static (vacuum) data set $(\Sigma;\sg,N)$ consists of an orientable three-manifold $\Sigma$, possibly with boundary, a Riemannian metric $\sg$, and a function $N$, such that,
\begin{enumerate}[labelindent=\parindent, leftmargin=*, label={\rm (\roman*)}, widest=a, align=left]
\item $N$ is strictly positive in the interior $\Sigma^{\circ}(=\Sigma\setminus \partial\Sigma)$ of $\Sigma$,
\item $(\sg,N)$ satisfy the vacuum static Einstein equations,
\be
\label{SEQ}  N Ric = \nabla\nabla N,\qquad \Delta N=0
\ee
\end{enumerate}
\end{Definition}
The definition is quite general. Observe in particular that $\Sigma$ and $\partial \Sigma$ could be compact or non-compact. To give an example, a data set $(\Sigma;\sg,N)$ can be simply the data inherited on any region of the Schwarzschild data. This flexibility in the definition of static data set allows us to write statements with great generality.

A horizon is defined as usual.
\begin{Definition}[Horizons] Let $(\Sigma;\sg,N)$ be a static vacuum data set. A horizon is a connected component of $\partial \Sigma$ where $N$ is identically zero.
\end{Definition}
Note that the Definition \ref{SDS} doesn't require $\partial \Sigma$ to be a horizon, though the data sets that we classify in this article are those with $\partial \Sigma$ consisting of a finite set of compact horizons ($\Sigma$ is a posteriori non compact). It is known that the norm $|\nabla N|$ is constant on any horizon and different from zero. It is called the surface gravity.

It is convenient to give a name to those spaces that are the final object of study of this article. Naturally we will call them {\it static black hole} data sets.

\begin{Definition}[Static black hole data sets]\label{DWO} A metrically complete static data set $(\Sigma;\sg,N)$ with $\partial \Sigma=\{N=0\}$ and $\partial \Sigma$ compact, is called a static black hole data set.
\end{Definition}  

The following definition, taken from \cite{4b6cb19bc94d4cf485e58571e3062f77}, recalls the notion of {\it weakly outermost} horizon.

\begin{Definition}[Galloway, \cite{4b6cb19bc94d4cf485e58571e3062f77}] Let $(\Sigma; \sg, N)$ be a static black hole data set. Then, a horizon $H$ is said weakly outermost if there are no embedded surfaces $S$ homologous to $H$ having negative outwards mean curvature. 
\end{Definition}

The following is the definition of Kasner asymptotic. It requires a decay into a background Kasner space faster than any inverse power of the distance. The definition follows the intuitive notion and it is written in the coordinates of the background Kasner, very much in the way AF is written in Schwarzschildian coordinates.

\begin{Definition}[Kasner asymptotic]\label{KADEF} A data set $(\Sigma; g,N)$ is asymptotic to a Kasner data $(\Sigma^{\mathbb{K}};g^{\mathbb{K}},N^{\mathbb{K}})$, $\Sigma_{\mathbb{K}}=(0,\infty)\times \T^{2}$, if for any $m\geq 1$ and $n\geq 0$ there is $C>0$, a bounded set $K\subset \Sigma$ and a diffeomorphism into the image $\phi:\Sigma\setminus K\rightarrow \Sigma_{\mathbb{K}}$ such that,
\begin{gather}
|\partial_{I}(\phi_{*}g)_{ij}-\partial_{I}g^{\mathbb{K}}_{ij}|\leq \frac{C}{x^{m}}\\
|\partial_{I}(\phi_{*}N)-\partial_{I}N^{\mathbb{K}}|\leq \frac{C}{x^{m}}
\end{gather}
for any multi-index $I=(i_{1},i_{2},i_{3})$ with $|I|=i_{1}+i_{2}+i_{3}\leq n$, where, if $x, y$ and $z$ are the coordinates in the Kasner space, then $\partial_{I}=\partial_{x}^{i_{1}}\partial_{y}^{i_{2}}\partial_{z}^{i_{3}}$. 
\end{Definition}
The next is the definition of data set of Myers/Korotkin-Nicolai type that we use.
\begin{Definition}[Black holes of M/KN type]\label{KNTDEF} A static-black hole data set $(\Sigma;\sg,N)$ is of Myers/Korotkin-Nicolai type if
\begin{enumerate}
\item $\partial \Sigma$ consist of $h\geq 1$ weakly outermost (topologically) spherical horizons,
\item $\Sigma$ is diffeomorphic to a solid three-torus minus  $h$-open three-balls,
\item the asymptotic is Kasner.
\end{enumerate}
\end{Definition}

It is worth to restate now the main classification theorem that we shall prove

\begin{Theorem}[The classification Theorem]\label{TCTHM2} Any static black hole data set is either,
\begin{enumerate}[labelindent=\parindent, leftmargin=*, label={\rm (\Roman*)}, widest=a, align=left]
\item\label{FTII} a Schwarzschild black hole, or,
\item\label{FTI} a Boost, or,
\item\label{FTIII} is of Myers/Korotkin-Nicolai type.
\end{enumerate}
\end{Theorem}

As an outcome of the proof (see Part II) it will be shown that the Kasner asymptotic of the static black holes of type \ref{FTIII}, that is of M/KN type, is different from the Kasner $A$ and $C$ (of course, as explained earlier, it can't be asymptotic to a Kasner with $\gamma<0$ by the maximum principle). We leave it as an open problem to prove that the only static black hole data sets asymptotic to $B$ are the Boosts. 

\begin{Problem} Prove that the Boosts are the only static black hole data sets asymptotic to a Boost.
\end{Problem}

It is also not known if the only static vacuum black holes of type \ref{FTIII} are the Myers/Korotkin-Nicolai static black holes. We state this as an open problem.

\begin{Problem}\label{OPENPRO} Prove (or disprove) that the only static vacuum black holes of type \ref{FTIII} are the Myers/Korotkin-Nicolai black holes.
\end{Problem}

On a large part of the article we will use the variables $(\hg,U)$ with $\hg=N^{2}g$ and $U=\ln N$, instead of the natural variables $(g,N)$. The data $(\Sigma;\hg,U)$ is the {\it harmonic presentation} of the data $(\Sigma;g,N)$. The static equations in these variables are,
\begin{gather}
Ric_{\hg}=2\nabla U\nabla U,\quad \Delta_{\hg} U=0
\end{gather}
and therefore the map $U:(\Sigma;\hg)\rightarrow \mathbb{R}$ is harmonic, (hence the name).

\subsection{Metric balls, annuli and partitions}\label{SAP}

\begin{enumerate}[leftmargin=*, label={\rm \arabic*}, widest=a, align=left]

\item {\sc Metric balls}. If $C$ is a set and $p$ a point then $\dist_{g}(C,p)=\inf\{\dist_{g}(q,p):q\in C\}$. Very often we take $C=\partial \Sigma$. If $C$ is a set and $r>0$, then, define the open ball of `center' $C$ and radius $r$ as,
\be
B_{g}(C,r)=\{p\in \Sigma:\dist_{g}(C,p)<r\}
\ee

\item {\sc Annuli}. Let $(\Sigma;g)$ be a metrically complete and non-compact Riemannian manifold with non-empty boundary $\partial \Sigma$. 

- Let $0<a<b$, then we define the open annulus $\mathcal{A}_{g}(a,b)$ as
\be
\mathcal{A}_{g}(a,b)=\{p\in \Sigma: a<\dist_{g}(p,\partial \Sigma)<b\}
\ee
We write just $\mathcal{A}(a,b)$ when the Riemannian metric $g$ is clear from the context. 

- If $C$ is a connected set included in $\mathcal{A}_{g}(a,b)$, then we write,
\be
\mathcal{A}^{c}_{g}(C;a,b)
\ee
to denote the connected component of $\mathcal{A}_{g}(a,b)$ containing $C$. The set $C$ could be for instance a point $p$ in which case we write $\mathcal{A}^{c}_{g}(p;a,b)$.

\item {\sc Partitions cuts and end cuts}. To understand the asymptotic geometry of data sets, we will study the geometry of scaled annuli. Sometimes however it will be more convenient and transparent to use certain sub-manifolds instead of annuli. For this purpose we define partitions, partition cuts, end cuts, and simple end cuts.

{\it Assumption}: Below we assume that $(\Sigma;g)$ is a metrically complete and non-compact Riemannian manifold with non-empty and compact boundary $\partial \Sigma$. 

\begin{Definition}[Partitions] A set of connected compact submanifolds of $\Sigma$ with non-empty boundary 
\be
\{\mathcal{P}^{m}_{j,j+1},\  j=j_{0},j_{0}+1,\ldots;\ m=1,2,\ldots,m_{j}\geq 1\}, 
\ee
($j_{0}\geq 0$), is a {\it partition} if,
\begin{enumerate}
\item $\mathcal{P}^{m}_{j,j+1}\subset \mathcal{A}(2^{1+2j},2^{4+2j})$ for every $j$ and $m$.
\item $\partial \mathcal{P}^{m}_{j,j+1}\subset (\mathcal{A}(2^{1+2j},2^{2+2j})\cup \mathcal{A}(2^{3+2j},2^{4+2j}))$ for every $j$ and $m$.
\item The union $\cup_{j,m}\mathcal{P}^{m}_{j,j+1}$ covers $\Sigma\setminus B(\partial \Sigma,2^{2+2j_{0}})$.
\end{enumerate}
\end{Definition}

\begin{figure}[h]
\centering
\includegraphics[width=7cm, height=9cm]{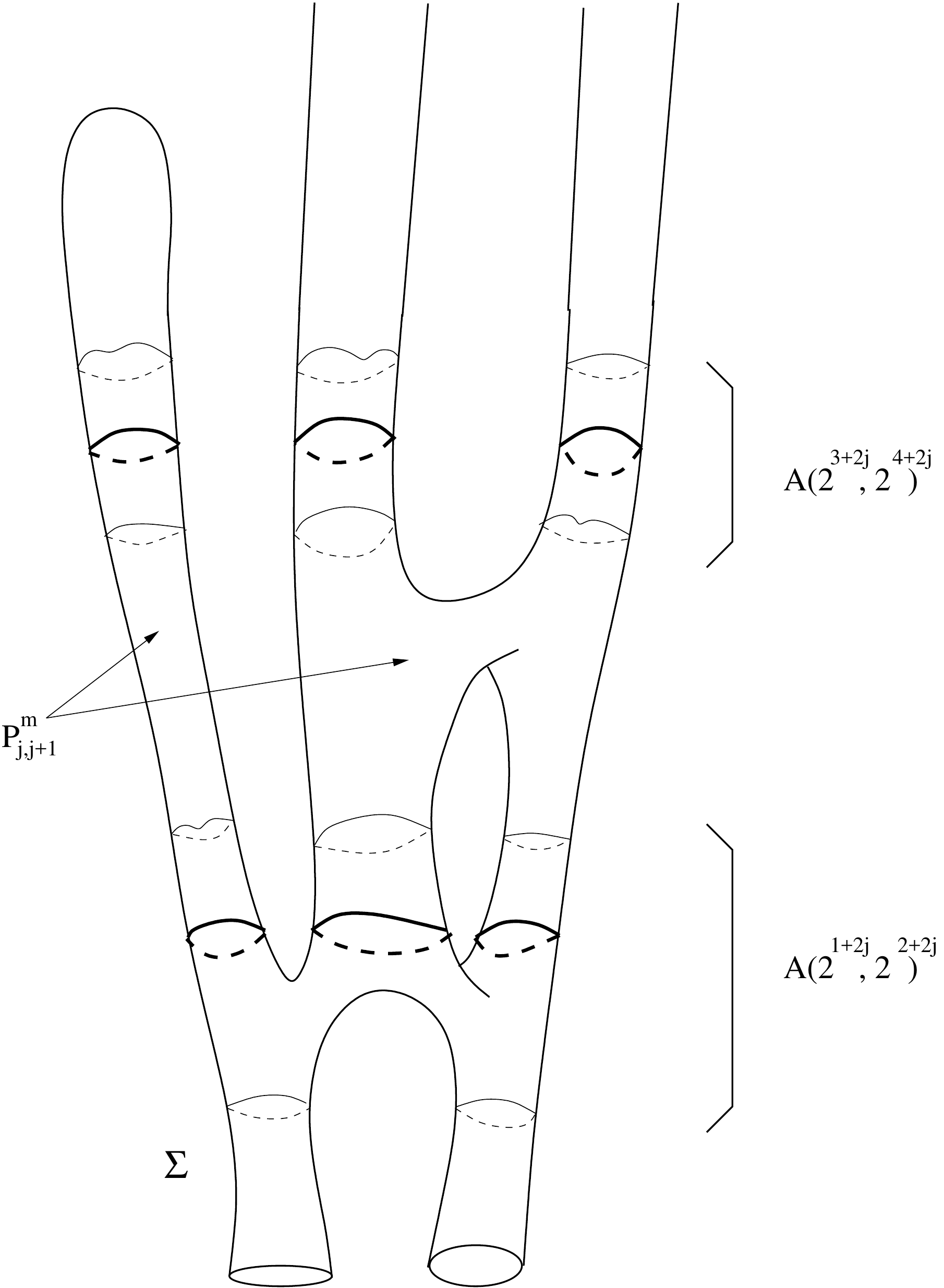}
\caption{The figure shows the annuli $\mathcal{A}(2^{1+2j},2^{2+2j})$, $\mathcal{A}(2^{3+2j},2^{4+2j})$ and the two components, for $m=1,2$ of $\mathcal{P}^{m}_{j,j+1}$.}
\label{PARTITIONF}
\end{figure}

Figure \ref{PARTITIONF} shows schematically a partition. The existence of partitions is done (succinctly) as follows. 
Let $j_{0}\geq 0$ and let $j\geq j_{0}$. Let $f:\Sigma\rightarrow [0,\infty)$ be a (any) smooth function such that $f\equiv 1$ on $\{p:\dist(p,\partial \Sigma)\leq 2^{1+2j}\}$ and $f\equiv 0$ on $\{p: \dist(p,\partial \Sigma)\geq 2^{2+2j}\}$, 
\footnote{Consider a partition of unity $\{\chi_{i}\}$ subordinate to a cover $\{\mathcal{B}_{i}\}$ where the neighbourhoods $\mathcal{B}_{i}$ are small enough that if $\mathcal{B}_{i}\cap \{p:\dist(p,\partial \Sigma)\leq 2^{1+2j}\}\neq \emptyset$ then $\mathcal{B}_{i}\cap \{p: \dist(p,\partial \Sigma)\geq 2^{2+2j}\}=\emptyset$. Then define $f=\sum_{i\in I}\chi_{i}$, where $i\in I$ iff $\mathcal{B}_{i}\cap \{p:\dist(p,\partial \Sigma_{i})\leq 2^{1+2j}\}\neq \emptyset$.}. 
Let $x$ be any regular value of $f$ in $(0,1)$. For each $j$ let $\mathcal{Q}_{j}$ be the compact manifold obtained as the union of the closure of the connected components of $\Sigma\setminus \{f=x\}$ containing at least a component of $\partial \Sigma$. Then the manifolds $\mathcal{P}^{m}_{j,j+1}$, $m=1,\ldots,m_{j}$, are defined as the connected components of $\mathcal{Q}_{j+1}\setminus \mathcal{Q}_{j}^{\circ}$. 

We let $\partial^{-}\mathcal{P}^{m}_{j,j+1}$ be the union of the connected components of $\partial \mathcal{P}^{m}_{j,j+1}$ contained in $\mathcal{A}(2^{1+2j},2^{2+2j})$. Similarly, we let $\partial^{+}\mathcal{P}^{m}_{j,j+1}$ be the union of the connected components of $\partial \mathcal{P}^{m}_{j,j+1}$ contained in $\mathcal{A}(2^{3+2j},2^{4+2j})$. 
\begin{Definition}[Partition cuts]
If $\mathcal{P}$ is a partition, then for each $j$ we let 
\be
\{\mathcal{S}_{jk},k=1,\ldots,k_{j}\} 
\ee
be the set of connected components of the manifolds $\partial^{-}\mathcal{P}^{m}_{j,j+1}$ for $m=1,\ldots,m_{j}$. The set of surfaces $\{\mathcal{S}_{jk},j\geq j_{0},\ldots,k=1,\ldots,k_{j}\}$ is called a {\it partition cut}. 
\end{Definition}

\begin{Definition}[End cuts] 
Say $\Sigma$ has only one end. Then, a subset, $\{\mathcal{S}_{jk_{l}}, l=1,\ldots,l_{j}\}$ of a partition cut
 $\{\mathcal{S}_{jk},k=1,\ldots,k_{j}\}$ is called an end cut if when we remove all the surfaces $\mathcal{S}_{jk_{l}}$, $l=1,\ldots,l_{j}$, from $\Sigma$, then every connected component of $\partial \Sigma$ belongs to a bounded component of the resulting manifold, whereas if we remove all but one of the surfaces $\mathcal{S}_{jk_{l}}$, then at least one connected component of $\partial \Sigma$ belongs to an unbounded component of the resulting manifold.
\end{Definition}

If $\Sigma$ has only one end, then one can always remove if necessary manifolds from a partition cut $\{\mathcal{S}_{jk},k=1,\ldots,k_{j}\}$ to obtain an end cut.
 
\begin{Definition}[Simple end cuts]
Say $\Sigma$ has only one end. If an end cut $\{\mathcal{S}_{jk_{l}},j\geq j_{0},l=1,\ldots,l_{j}\}$ has $l_{j}=1$ for each $j\geq j_{0}$ then we say that the end is a {\it simple end cut} and write simply $\{\mathcal{S}_{j}\}$.
\end{Definition}

If $\{\mathcal{S}_{j}\}$ is a simple end cut and $j_{0}\leq j<j'$ we let $\mathcal{U}_{j,j'}$ be the compact manifold enclosed by $\mathcal{S}_{j}$ and $\mathcal{S}_{j'}$. This notation will be used very often.

\end{enumerate}

\subsection{A Harnak-type of estimate for the Lapse}\label{TBCPHTE1}

Let $(\Sigma;\sg,N)$ be a metrically complete static data set with $\partial \Sigma$ compact. In \cite{MR1809792}, Anderson observed that, as the four-metric $N^{2}dt^{2}+g$ is Ricci-flat, then Liu's ball-covering property holds \cite{MR1216638} (the compactness of $\partial \Sigma$ is necessary here because Liu's theorem is for manifolds with non-negative Ricci curvature outside a compact set). Namely, for any $b>a>\delta>0$ there is $n$ and $r_{0}$ such that for any $r\geq r_{0}$ the annulus $\mathcal{A}(ra,rb)$ can be covered by at most $n$ balls of $g$-radius $r\delta$ centred in the same annulus. Hence any two points $p$ and $q$ in a connected component of $\mathcal{A}(ra,rb)$ can be joined through a chain, say $\alpha_{pq}$, of at most $n+2$ radial geodesic segments of the balls of radius $\delta$ covering $\mathcal{A}(ra,rb)$. On the other hand Anderson's estimate (see subsection \ref{CMACD}) implies that the $g$-gradient $|\nabla \ln N|_{r}$ is bounded by $C/r$. 
Integrating $|\nabla \ln N|$ along the curves $\alpha_{pq}$ and using Anderson's bound we arrive at a relevant Harnak estimate controlling uniformly the quotients $N(p)/N(q)$. The estimate is due to Anderson and is summarised in the next Proposition (for further details see, \cite{0264-9381-32-19-195001}).
 
\begin{Proposition}{\rm (Anderson, \cite{MR1809792})}\label{MAXMINU11} Let $(\Sigma;g,N)$ be a metrically complete static data set with $\partial \Sigma$ compact and let $0<a<b$. Then, there is $r_{0}$ and $\eta>0$, such that for any $r>r_{0}$ and for any set $Z$ included in a connected component of $\mathcal{A}(a,b)$ we have,
\be\label{EQHARN1}
\max\{N(p):p\in Z\}\leq \eta \min\{N(p):p\in Z\}
\ee
\end{Proposition}

\section{Conformal transformations by powers of the lapse}\label{CTPL}

In this section we study conformal transformations of static metrics by powers of the lapse from the point of view \'a la Backry-\'Emery. The contents are the following.

Subsection \ref{BESEC} explains the structure of the conformal equations, Proposition \ref{FELIZ}. Subsection \ref{CMACD} proves Lemma \ref{LEMMAME} and from it its is obtained a generalised Anderson's decay estimate for the conformally related data, Lemma \ref{CDLEMMA}. These estimates are used in subsection \ref{CMMC} to show the metric completeness of the manifolds $(\Sigma; \overline{g}=N^{-2\epsilon}g)$ for $-1-\sqrt{2}<\epsilon<-1+\sqrt{2}$ (provided $\partial \Sigma$ is compact and $N|_{\Sigma}>0$ and $(\Sigma;g)$ is metrically complete), Theorem \ref{COMN2}. Subsection \ref{APP} contains important applications. First, in subsection \ref{CDPL} a few important remarks are pointed out on the conformal data $(\Sigma; N^{-2\epsilon}g)$ of a static data $(\Sigma;g)$, Proposition \ref{PIV}. It is particularly stressed here that when $\epsilon>0$ is small, the manifold $(\Sigma; N^{-2\epsilon}g)$ is still metrically complete, while the boundary becomes strictly convex (indeed the boundary of $\Sigma$ minus a small collar around $\partial \Sigma$). In subsection \ref{STR} it is proved using the previous subsection and a generalised splitting theorem \'a la Backry-\'Emery that static black hole data sets have only one end, Proposition \ref{KUNO}. In subsection \ref{HTT} it is proved using the completeness at infinity of $(\Sigma; N^{2}g=\hg)$ that either black holes data sets are boosts, or every horizon component is a sphere and is weakly outermost. Finally in subsection \ref{TAIS} it is proved that static isolated systems in GR are asymptotically flat. This application is independent of the rest of the article.

\subsection{Conformal metrics, the Bakry-\'Emery Ricci tensor and the static equations}\label{BESEC}
Given a Riemannian metric $g$, function $f$ and constant $\alpha$, the $\alpha$-Bakry-\'Emery Ricci tensor $Ric^{\alpha}_{f}$ is defined as (see \cite{MR2577473}; note that \cite{MR2577473} uses the notation $1/N$ instead of $\alpha$),
\be
Ric^{\alpha}_{f}:=Ric+\nabla\nabla f-\alpha\nabla f \nabla f,\\
\ee
where the tensors $Ric$ and $\nabla$ on the right hand side are with respect to $g$. The $f$-Laplacian $\Delta_{f}$ acting on a function $\phi$ is defined as
\be
\Delta_{f}\phi:=\Delta\phi-\langle \nabla f,\nabla \phi\rangle
\ee
where again $\Delta$ on the right hand side are with respect to $g$ and $\langle\ ,\ \rangle=g(\ ,\ )$. Now observe that letting $f:=-\ln N$, the static Einstein equations (\ref{SEQ}) read
\be\label{RICCILN}
Ric = -\nabla\nabla f +\nabla f\nabla f,\qquad \Delta f - \langle \nabla f,\nabla f\rangle =0
\ee
In the notation above, this is nothing else than to say that
\be
Ric^{\alpha}_{f}=0,\qquad \Delta_{f} f=0
\ee
with $\alpha=1$ and $f=-\ln N$. It is an important fact that the structure of these equations is preserved along a one parameter family of conformal transformations. The following calculation explains this fact.
\begin{Proposition}\label{FELIZ} Let $(\sM; g,N)$ be a static data set. Fixed $\epsilon$ define  
\be
\overline{\sg}=N^{-2\epsilon}g.
\ee
Then, 
\be\label{OOO}
\overline{Ric}^{\alpha}_{f}=0,\qquad \overline{\Delta}_{f} f=0
\ee
where $\alpha=(1-2\epsilon-\epsilon^{2})/(1+\epsilon)^{2}$ and $f=-(1+\epsilon)\ln N$.
\end{Proposition}

We used the notation $\overline{Ric}$ for $Ric_{\overline{g}}$ and $\overline{\Delta}$ for $\Delta_{\overline{g}}$.

Note that when $\epsilon=-1$, we obtain $\alpha=+\infty$, $f=0$ and $\overline{Ric}^{\alpha}_{f}=\overline{Ric}-2\nabla \ln N \nabla \ln N$. In particular we recover $\overline{Ric}=2\nabla \ln N\nabla \ln N$.

\begin{proof} We prove first $\overline{\Delta}_{f} f=0$. Recall from standard formulae that if $\overline{g}=e^{2\psi}g$ then for every $\phi$ we have
\be\label{EQU}
e^{-2\psi}\Delta \phi = \overline{\Delta} \phi -\langle \nabla \phi,\nabla \psi\rangle_{\overline{g}}
\ee
Making $\phi=\ln N$ and $e^{\psi}=N^{-\epsilon}$, the left hand side of (\ref{EQU}) is equal to $-|\nabla \ln N|^{2}_{\overline{g}}$ because $\Delta \ln N=-|\nabla \ln N|^{2}_{g}$. Thus (\ref{EQU}) is $\overline{\Delta} \ln N - \langle \nabla \ln N, - (1 +\epsilon)\nabla \ln N\rangle_{\overline{g}}=0$ as wished.

Let us prove now $\overline{Ric}^{\alpha}_{f}=0$. Recall first that if $\overline{g}=e^{2\psi}g$ then 
\be
\overline{Ric}=Ric-(\nabla\nabla \psi-\nabla\psi\nabla\psi)-(\Delta\psi+|\nabla \psi|^{2})g
\ee
Choosing $\psi=-\epsilon\ln N$ and replacing $Ric$ by (\ref{RICCILN}) then  gives
\be\label{FAFA}
\overline{Ric}=(1+\epsilon)\nabla\nabla \ln N+(1+\epsilon^{2})\nabla\ln N\nabla \ln N-(\epsilon+\epsilon^{2})|\nabla \ln N|^{2}g
\ee
Use now the usual general formula
\be
\overline{\nabla}_{i}V_{j}=\nabla_{i}V_{j}-\big[V_{j}\nabla_{i}\psi+V_{i}\nabla_{j}\psi-(V^{k}\nabla_{k}\psi)g_{ij}\big]
\ee
with $V^{j}=\nabla_{j} \ln N$ and with $\psi=-\epsilon\ln N$, to obtain
\be\label{FAF}
\nabla\nabla \ln N=\overline{\nabla}\nabla \ln N -\epsilon \big[2\nabla\ln N\nabla \ln N - |\nabla \ln N|^{2} g\big]
\ee
Plugging (\ref{FAF}) in (\ref{FAFA}) gives
\be
\overline{Ric}=(1+\epsilon)\overline{\nabla}\nabla \ln N+(1-2\epsilon-\epsilon^{2})\nabla \ln N\nabla \ln N
\ee
which is $\overline{Ric}^{\alpha}_{f}=0$ as claimed. 
\end{proof}

\subsection{Conformal metrics and Anderson's curvature decay}\label{CMACD}

In \cite{MR1806984} Anderson proved the following fundamental quadratic curvature decay for static data sets. 

\begin{Lemma}[Anderson, \cite{MR1806984}]\label{LACD1}
There is a constant $\eta>0$ such that for any metrically complete static data set $(\Sigma;g,N)$ we have,
\be\label{CURVDEC}
|Ric|(p)\leq \frac{\eta}{\dist^{2}(p,\partial \Sigma)},\qquad |\nabla \ln N|^{2}(p)\leq \frac{\eta}{\dist^{2}(p,\partial \Sigma)},
\ee
for any $p\in \Sigma^{\circ}$. 
\end{Lemma}
This decay estimate is linked to a similar one for the metric $\hg=N^{2}\sg$ that we state below. It was proved also by Anderson in \cite{MR1806984}. We require $N>0$ everywhere and not only on $\Sigma^{\circ}$, to guarantee that $\hg$ is regular on $\partial \Sigma$. Note that imposing $N>0$ on $\Sigma$, does not make $(\Sigma;\hg=N^{2}\sg)$ automatically metrically complete. Indeed if $\Sigma$ is non-compact then $N$ could tend to zero over a divergent sequence of points and this may cause the metric incompleteness of the space $(\Sigma;\hg)$. 
\begin{Lemma}[Anderson \cite{MR1806984}]\label{LACD2}
There is a constant $\eta>0$ such that, for any static data set $(\Sigma;g,N)$ with $N>0$ and for which $(\Sigma;\hg=N^{2}\sg)$ is metrically complete, we have
\be\label{CURVDEC2}
|Ric_{\hg}|_{\hg}(p)\leq \frac{\eta}{\dist^{2}_{\hg}(p,\partial \Sigma)},\qquad |\nabla \ln N|_{\hg}^{2}(p)\leq \frac{\eta}{\dist^{2}_{\hg}(p,\partial \Sigma)}
\ee
for any $p\in \Sigma^{\circ}$. 
\end{Lemma} 

The estimates (\ref{CURVDEC}) and (\ref{CURVDEC2}) are particular instances of a whole family of estimates for the conformal metrics $\overline{g}=N^{-2\epsilon}g$, with $\epsilon$ ranging in the interval $(-1-\sqrt{2},-1+\sqrt{2})$ which is the interval where the polynomial $1-2\epsilon-\epsilon^{2}$ is positive. We prove the estimates below using the results in Section \ref{BESEC}. As a byproduct we  provide concise proofs of Lemmas \ref{LACD1} and \ref{LACD2}. This will be the goal of this section.

We start with a lemma that to our knowledge is essentially due to J. Case \cite{MR2741248} (though similar techniques are well known too at least in the theory of minimal surfaces). This lemma was first presented in \cite{Reiris2017}, but due to its importance we prove it again here.

\begin{Lemma}\label{LEMMAME} Let $(\Sigma,g)$ be a metrically complete Riemannian three-manifold with $Ric^{\alpha}_{f}\geq 0$ for some function $f$ and constant $\alpha>0$. Let $\phi$ be a non-negative function such that
\be\label{FUNLAP}
\Delta_{f}\phi\geq c\phi^{2}
\ee
for some constant $c>0$. Then, for any $p\in \Sigma^{\circ}$ we have
\be\label{FUNDEST}
\phi(p)\leq \frac{\eta}{\dist^{2}(p,\partial \Sigma)}
\ee
where $\eta=(36+4/\alpha)/c$.  
\end{Lemma}

Observe that the lemma applies too to manifolds with $Ric\geq 0$ as this corresponds to the case $Ric_{f=0}^{\alpha}\geq 0$ for any $\alpha>0$.

\begin{proof} For any function $\chi$ the following general formula holds
\be
\Delta_{f}(\chi\phi)=\phi(\Delta_{f} \chi)+2\langle \nabla\chi,\nabla \phi\rangle +\chi\Delta_{f}\phi
\ee
Thus, if $\chi\geq 0$ and if $q$ is a local maximum of $\chi\phi$ on $\Sigma^{\circ}$, we have
\be\label{PRELCALC}
0\geq \bigg[\Delta_{f}(\chi\phi)\bigg]\bigg|_{q} \geq  \bigg[\phi \Delta_{f}\chi - 2\frac{|\nabla \chi|^{2}}{\chi}\phi +c\chi\phi^{2}\bigg]\bigg|_{q}
\ee
where to obtain the second inequality we used (\ref{FUNLAP}). Let $r_p=\dist(p,\partial \Sigma)$.  On $B(p,r_p)$ let the function $\chi(x)$ be $\chi(x)=(r_p^{2}-r(x)^{2})^{2}$. To simplify notation make $r=r(x)=\dist(x,p)$. Let $q$ be a point in the closure of $B(p,r_p)$ where the maximum of $\chi\phi$ is achieved. If $\phi(q) = 0$, then $\phi = 0$ and (\ref{FUNDEST}) holds for any $\eta>0$. So let us assume that $\phi(q)>0$. In particular $p$ belongs to the interior of $B(p,r_p)$. By (\ref{PRELCALC}) we have
\begin{align}\label{PREVIOUS}
cr_p^{4}\phi(p) \leq c(\chi\phi)(q) & \leq \bigg[2\frac{|\nabla\chi|^{2}}{\chi}-\Delta_{f}\chi \bigg]\bigg|_{q}\\
& =\bigg[4(r_p^{2}-r^{2})r\Delta_{f}r+4r_p^{2}+20r^{2}\bigg]\bigg|_{q}
\end{align}
But if $Ric_{f}^{\alpha}\geq 0$ then $\Delta_{f} r\leq (3+1/\alpha)/r$, (see \cite{MR2577473} Theorem A.1; On non-smooth points of $r$ this equations holds in the barrier sense\footnote{This is an important property as it allows us to make analysis as if $r$ were a smooth function, see \cite{MR2243772}.}). Using this in (\ref{PREVIOUS}) and after a simple computation we deduce,
\be
\phi(p)\leq \frac{(4(3+1/\alpha)+24)}{c r_p^{2}},
\ee 
which is (\ref{FUNDEST}).
\end{proof}

Let us see now an application of the previous Lemma. Let $(\Sigma;g,N)$ be a static data with $N>0$. Let $\epsilon$ be a number in $(-1-\sqrt{2},-1+\sqrt{2})$ and assume that the space ($\Sigma$; $\overline{g}=N^{-2\epsilon}g$) is metrically complete. We claim that there is $\eta(\epsilon)>0$, such that for all $p\in \Sigma^{\circ}$ we have
\be\label{AESTS}
|\nabla \ln N|^{2}_{\overline{g}}(p)\leq \frac{\eta(\epsilon)}{\dist^{2}_{\overline{g}}(p,\partial \sM)}
\ee  
Let us prove the claim. Assume first $\epsilon\neq -1$. From Lemma \ref{LEMMAME} we know that $\overline{Ric}^{\alpha}_{f}=0$ where $f=-(1+\epsilon)\ln N$ and where $\alpha=(1-2\epsilon-\epsilon^{2})/(1+\epsilon)^{2}$. The factor $(1-2\epsilon-\epsilon^{2})$ is greater than zero by the assumption on the range of $\epsilon$. Now use the general formula (see\cite{MR2741248})
\be\label{BOCHNERF}
\frac{1}{2}\overline{\Delta}_{f} |\nabla \phi|_{\overline{g}}^{2} = |\overline{\nabla}\nabla \phi|_{\overline{g}}^{2}+\langle\nabla \phi,\nabla(\overline{\Delta}_{f}\phi)\rangle_{\overline{g}} +\overline{Ric}^{\alpha}_{f}(\nabla\phi, \nabla \phi) +\alpha\langle \nabla f,\nabla \phi\rangle_{\overline{g}}^{2}
\ee
with $\phi=\ln N$, together with $\overline{Ric}^{\alpha}_{f}=0$, to obtain
\be
\overline{\Delta}_{f}  |\nabla \ln N|_{\overline{g}}^{2}\geq 2(1-2\epsilon-\epsilon^{2})|\nabla \ln N|_{\overline{g}}^{4}
\ee
and thus (\ref{AESTS}) from Lemma \ref{LEMMAME}. When $\epsilon=-1$ then $\overline{Ric}^{\alpha}_{f=0}\geq 0$ for any $\alpha>0$ and 
\be
\overline{\Delta}_{f=0}  |\nabla \ln N|_{\overline{g}}^{2}\geq 4|\nabla \ln N|_{\overline{g}}^{4}
\ee
The claim again follows from Lemma \ref{LEMMAME}.

Note that Lemma \ref{LEMMAME} provides the following explicit expression for $\eta(\epsilon)$,
\be
\eta(\epsilon)= \frac{1}{2(1-2\epsilon-\epsilon^{2})}\bigg[36+\frac{4(1+\epsilon)^{2}}{(1-2\epsilon-\epsilon^{2})}\bigg]
\ee

What we just showed is a part of the {\it generalised Anderson's quadratic curvature decay} mentioned earlier, that we now state and prove. 

\begin{Lemma}\label{CDLEMMA} 
Let $\epsilon$ be a number in the interval $(-1-\sqrt{2},-1+\sqrt{2})$. Then there is $\eta(\epsilon)$ such that for any static data set $(\Sigma;g,N)$ with $N>0$ and for which $(\Sigma; \overline{g}=N^{-2\epsilon}\sg)$ is metrically complete, we have,
\be\label{CDLEMMAEST}
|\overline{Ric}|_{\overline{g}}(p)\leq \frac{\eta(\epsilon)}{\dist^{2}_{\overline{g}}(p,\partial \Sigma)}, \qquad |\nabla \ln N|^{2}_{\overline{g}}(p)\leq \frac{\eta(\epsilon)}{\dist^{2}_{\overline{g}}(p,\partial \Sigma)},
\ee
for any $p\in \Sigma^{\circ}$.
\end{Lemma}

\begin{proof} We have already shown the second estimate of (\ref{CDLEMMAEST}). If $\partial \Sigma=\emptyset$ then $N$ is constant and $\overline{g}$ is flat. So let us assume that $\partial \Sigma\neq \emptyset$. Let $p\in \Sigma^{\circ}$. By scaling we can assume without loss of generality that $N(p)=1$ and $\overline{d}_{p}=\dist_{\overline{g}}(p,\partial \Sigma)=1$. In this setup, we need to prove that 
\be\label{RMNN}
|\overline{Ric}|_{\overline{\sg}}(p)\leq c_{0}(\epsilon), 
\ee
for $c_{0}$ independent of the data.

The second estimate of (\ref{CDLEMMAEST}) yields, 
\be\label{RAPA1}
|\nabla \ln N|_{\overline{g}}(x)\leq c_{1},\\
\ee
for all $x\in B_{\overline{g}}(p,1/2)$ and where $c_{1}=c_{1}(\epsilon)$ is independent of the data. Therefore, as,
\be
\overline{Ric}=(1+\epsilon)\overline{\nabla}\nabla \ln N+(1-2\epsilon-\epsilon^{2})\nabla \ln N\nabla \ln N,
\ee
then to prove (\ref{RMNN}) it is enough to prove 
\be\label{CARTOON3}
|\overline{\nabla}\nabla \ln N|_{\overline{\sg}}(p)\leq c'_{0}(\epsilon)
\ee
for a $c'_{0}(\epsilon)$ independent of the data.

Let $\gamma(s)$ be a geodesic segment joining $p$ to $x$. Then we can write,
\be
\big|\ln \frac{N(x)}{N(p)}\big|=\big|\int \nabla_{\gamma'}\ln Nds\big|\leq \int |\nabla\ln N|_{\overline{\sg}}ds\leq c_{1}/2
\ee
where we used (\ref{RAPA1}). Because $N(p)=1$, this inequality gives,
\be\label{RAPA2}
0<c_{2}\leq N(x)\leq c_{3}<\infty
\ee
for all $x\in B_{\overline{g}}(p,1/2)$ and where $c_{2}=c_{2}(\epsilon)$ and $c_{3}=c_{3}(\epsilon)$. 

Let $\hg=N^{2+2\epsilon}\overline{\sg}=N^{2}g$. If $\epsilon\geq -1$ let $r_{0}=c_{2}^{1+\epsilon}$, whereas if $\epsilon<-1$ let $r_{0}=c_{3}^{1+\epsilon}$. Then, clearly $B_{\hg}(p,r_{0})\subset B_{\overline{g}}(p,1/2)$. Moreover (\ref{RAPA1}) and (\ref{RAPA2}) show that for all $x\in B_{\hg}(p,r_{0})$ we have,
\be\label{CARTOON1}
|\nabla \ln N|_{\hg}(x)\leq c_{4}(\epsilon),\\
\ee
As $Ric_{\hg}=2\nabla \ln N \nabla \ln N$, we deduce that 
\be
|Ric_{\hg}|_{\hg}(x)\leq c_{5}(\epsilon)
\ee
for all $x\in B_{\hg}(p,r_{0})$. In dimension three the Ricci tensor determines the Riemann tensor, so, 
\be\label{RARA2}
|Rm_{\hg}|_{\hg}(x)\leq c_{6}(\epsilon)
\ee
Hence, by standard arguments, there is $r_{1}(\epsilon)\leq r_{0}$ such that the exponential map $exp:B^{\mathcal{T}}_{\hg}(p,r_{1})\rightarrow \Sigma$, is a diffeomorphism into the image, ($B_{\hg}^{\mathcal{T}}(p,r_{1})$ is a ball in $\mathcal{T}_{p}\Sigma$). Let $\tilde{\hg}$ be the lift of $\hg$ to $B^{\mathcal{T}}_{\hg}(p,r_{1})$ by $exp^{-1}$. We still have the bound (\ref{RARA2}) for $\tilde{\hg}$ and as the injectivity radius $inj_{\hg}(p)$ is bounded from below by $r_{1}$, then the {\it harmonic radius} $i_{h}(p)$, which controls the geometry in $C^{2}$ (see \cite{MR2243772}), is bounded from below by $r_{2}(\epsilon)\leq r_{1}$. As $\Delta_{\tilde{\hg}}\ln N=0$, then standard elliptic estimates give 
\be\label{CARTOON2}
|\nabla^{\tilde{\hg}} \nabla \ln N|_{\tilde{\hg}}(p)\leq c_{7}(\epsilon), 
\ee
where $\nabla^{\tilde{\hg}}$ is the covariant derivative of $\tilde{\hg}$. Finally, (\ref{RAPA2}), (\ref{CARTOON1}), (\ref{CARTOON2}) and the general formula,
\be
\overline{\nabla}\nabla\ln N=\nabla^{\hg}\nabla \ln N-(1+\epsilon)\big[2\nabla\ln N\nabla \ln N-|\nabla \ln N|^{2}_{\hg}\hg\big]
\ee 
provide the required bound (\ref{CARTOON3}). This completes the proof.
\end{proof}

It is easy to check using elliptic estimates that the proof of the Lemma (\ref{CDLEMMA}) leads also to the estimates 
\be\label{ESTCHEC}
|\overline{\nabla}^{(k)}\overline{Ric}|_{\overline{g}}(p)\leq \frac{\eta_{k}(\epsilon)}{\dist^{2+k}_{\overline{g}}(p,\partial \Sigma)}, \qquad |\overline{\nabla}^{(k)}\nabla \ln N|^{2}_{\overline{g}}(p)\leq \frac{\eta_{k}(\epsilon)}{\dist^{2+2k}_{\overline{g}}(p,\partial \Sigma)}
\ee
for every $k\geq 1$, where $\overline{\nabla}^{(k)}$ is $\overline{\nabla}$ applied $k$-times and where the positive constants $\eta(\epsilon)$, $\eta_{1}(\epsilon)$, $\eta_{2}(\epsilon)$, $\eta_{3}(\epsilon),\ldots$ are independent of the data set.

\subsection{Conformal metrics and metric completeness}\label{CMMC} 

In this section we aim to prove that metric completeness of data sets (with $N>0$ and $\partial \Sigma$ compact) imply the metric completeness of the conformal spaces $(\Sigma;\overline{\sg}=N^{-2\epsilon}\sg)$ for any $\epsilon$ in the range $(-1-\sqrt{2},-1+\sqrt{2})$. Note that until now, when it was necessary we have been including the completeness of the metrics $\overline{g}$ as a hypothesis.

\begin{Theorem}\label{COMN2} Let $\epsilon$ be a number in the interval $(-1-\sqrt{2},-1+\sqrt{2})$. Let $(\Sigma;g,N)$ be a metrically complete static data set with $N>0$ and $\partial \Sigma$ compact. Then $(\Sigma;\overline{\sg}=N^{-2\epsilon}\sg)$ is metrically complete.
\end{Theorem}

We start proving a corollary to Lemma \ref{CDLEMMA} that estimates $N$.

\begin{Corollary} {\rm (to Lemma \ref{CDLEMMA})} Let $\epsilon$ be a number in the interval $(-1-\sqrt{2},-1+\sqrt{2})$. Let $(\Sigma;g,N)$ be a static data set with $N>0$ and $\partial \Sigma$ compact, and for which $(\Sigma, \overline{g}=N^{-2\epsilon}\sg)$ is metrically complete. Then, there is $c>0$ (depending on the data) such that
\be\label{OBS}
\frac{1}{c(1+\dist_{\overline{\sg}}(p,\partial \Sigma))^{\sqrt{\eta}}}\leq N(p)\leq c(1+\dist_{\overline{\sg}}(p,\partial \Sigma))^{\sqrt{\eta}}
\ee
for any $p\in \Sigma^{\circ}$, where $\eta=\eta(\epsilon)$ is the coefficient in the decay estimate (\ref{CDLEMMAEST}) of Lemma \ref{CDLEMMA}. 
\end{Corollary}

\begin{proof} Let $p\in \Sigma$ such that $\overline{d}_{p}:=\dist_{\overline{\sg}}(p,\partial \Sigma)\geq 1$ (if it exists). Let $\gamma(\overline{s})$ be a $\overline{\sg}$-geodesic segment joining $\partial \Sigma$ to $p$ and realising the $\overline{\sg}$-distance between them (in particular $N(\gamma(\overline{d}_{p}))=N(p)$). Then we can write
\be
\bigg|\ln \frac{N(\gamma(\overline{d}_{p}))}{N(\gamma(1))}\bigg|=\bigg|\int_{1}^{\overline{d}_{p}}\nabla_{\gamma'}\ln Nd\overline{s}\bigg|\leq \int_{1}^{\overline{d}_{p}}\big|\nabla \ln N\big|d\overline{s}\leq \sqrt{\eta(\epsilon)}\ln \overline{d}_{p}
\ee
where to obtain the last inequality we have used (\ref{AESTS}). Therefore,
\be
N(p)\leq N(\gamma(1))\overline{d}_{p}^{\sqrt{\eta}}\quad {\rm and}\quad N(p)\geq N(\gamma(1))/\overline{d}_{p}^{\sqrt{\eta}}
\ee
Thus,
\be\label{OBSE}
\overline{m}\overline{d}_{p}^{\sqrt{\eta}}\geq N(p) \geq \underline{m}/d_{p}^{\sqrt{\eta}}
\ee
where $\overline{m}=\max\{N(q):\dist_{\overline{\sg}}(q,\partial \Sigma)=1\}$ and $\underline{m}=\min\{N(q):\dist_{\overline{\sg}}(q,\partial \Sigma)\}$. This clearly implies (\ref{OBS}). Obtaining (\ref{OBS}) for all $p\in \Sigma^{\circ}$, namely even for those with $\overline{d}_{p}\leq 1$, is direct due to the compactness of $\partial \Sigma$.
\end{proof}

\begin{Proposition}\label{PURURU} Let $\epsilon$ be a number in the interval $(-1-\sqrt{2},-1+\sqrt{2})$. Let $(\Sigma;g,N)$ be a static data set with $N>0$ and for which $(\Sigma, \overline{g}=N^{-2\epsilon}\sg)$ is metrically complete. Then, for any $\zeta$ such that $|\zeta|\leq 1/(2\sqrt{\eta})$, the space $(\Sigma;N^{2\zeta}\overline{\sg})$ is metrically complete, where $\eta=\eta(\epsilon)$ is the coefficient in (\ref{CDLEMMAEST}).
\end{Proposition}

\begin{proof} Let us assume that $\Sigma$ is non-compact otherwise there is nothing to prove. Let $\hat{\sg}=N^{2\zeta}\overline{\sg}$. To prove that $(\Sigma;\hat{\sg})$ is complete, we need to show that the following holds: for any sequence of points $p_{i}$ whose $\overline{g}$-distance to $\partial \Sigma$ diverges, then the $\hat{g}$-distance to $\partial \Sigma$ also diverges. Equivalently, we need to prove that for any sequence of curves $\alpha_{i}$ starting at $\partial \Sigma$ and ending at $p_{i}$ we have
\be
\int_{0}^{\overline{s}_{i}}N^{\zeta}(\alpha_{i}(\overline{s}))d\overline{s}\longrightarrow \infty
\ee
where $\overline{s}$ is the $\overline{\sg}$-arc length of $\alpha_{i}$ counting from $\partial \Sigma$. 

From (\ref{OBS}) we get,
\be
N^{\zeta}(p)\geq \frac{c^{-|\zeta|}}{(1+\dist_{\overline{\sg}}(p,\partial \Sigma))^{|\zeta|\sqrt{\eta}}}
\ee
for all $p$. But, $\dist_{\overline{\sg}}(\alpha_{i}(\overline{s}),\partial \Sigma)\leq \overline{s}$ and $|\zeta|\leq 1/(2\sqrt{\eta})$, so we deduce,
\be
N^{\zeta}(\alpha_{i}(\overline{s}))\geq \frac{c^{-|\zeta|}}{(1+\overline{s})^{1/2}}
\ee
Thus,
\be
\int_{0}^{\overline{s}_{i}}{N^{\zeta}(\alpha_{i}(\overline{s}))}d\overline{s} \geq \int_{0}^{\overline{s}_{i}} \frac{c^{-|\zeta|}}{(1+\overline{s})^{1/2}}d\overline{s}\longrightarrow \infty
\ee
as $\overline{s}_{i}\rightarrow \infty$ as wished. 
\end{proof}

We prove now Theorem \ref{COMN2}. 

\begin{proof}[Proof of Theorem \ref{COMN2}] Let $\epsilon\in (-1-\sqrt{2},-1+\sqrt{2})$. Assume $\epsilon\neq 0$ otherwise there is nothing to prove. Let $n>0$ be an integer such that for any $i=0,1,\ldots,n-1$, 
\be\label{SAYS}
\big|\frac{\epsilon}{n}\big|\leq \frac{1}{2\sqrt{\eta(i\epsilon/n)}}
\ee
where $\eta$ is the coefficient in (\ref{CDLEMMAEST}). According to Proposition \ref{PURURU}, the condition (\ref{SAYS}) says that if $\overline{\sg}_{i}=N^{-2(i\epsilon/n)}g$ is complete then so is $\overline{\sg}_{i+1}=N^{-2\epsilon/n}\overline{\sg}_{i}=N^{-2(i+1)\epsilon/n}g$ for any $i=0,1,\ldots,n-1$. Therefore, as $g$ is complete, then so are $\overline{\sg}_{1}$, $\overline{\sg}_{2}$, $\overline{\sg}_{3}$, until $\overline{\sg}_{n}=N^{-2\epsilon}g$ as wished.
\end{proof}

\subsection{Applications}\label{APP}

\subsubsection{Conformal transformations of black hole metrics}\label{CDPL} 
 
Let $(\Sigma;g,N)$ be a static black hole data set. We denote by $\Sigma_{\delta}$ the manifold resulting after removing from $\Sigma$ the $g$-tubular neighbourhood of $\partial \Sigma$ and radius $\delta$, i.e. $\Sigma_{\delta}=\Sigma\setminus B(\partial \Sigma,\delta)$. Let $\delta_{0}$ be small enough that $\partial \Sigma_{\delta}$ is always smooth and isotopic to $\partial \Sigma$ for any $\delta\leq\delta_{0}$.

Given $\epsilon>0$ let $\overline{g}=N^{-2\epsilon}g$. Let $\delta>0$ such that $\delta<\delta_{0}$. The second fundamental form $\overline{\Theta}$ of $\partial \Sigma_{\delta}$, (with respect to $\overline{g}$ and with respect to the inward normal to $\Sigma_{\delta}$), is 
\be\label{GOGO}
\overline{\Theta}=N^{\epsilon}\Theta -\epsilon\frac{\nabla_{n} N}{N^{1-\epsilon}}g
\ee
where $\Theta$ is the second fundamental form of $\partial \Sigma_{\delta}$ with respect to $g$ and $n$ is the inward $g$-unit normal. If we let $\delta\rightarrow 0$, the function $\nabla_{n}N|_{\partial \Sigma_{\delta}}$ converges (on each connected component) to a positive constant (the surface gravity) while $N|_{\partial \Sigma_{\delta}}$ converges to zero. Hence if $\delta$ is small enough, the second term on the right hand side of (\ref{GOGO}) dominates over the first, and the boundary $\partial \Sigma_{\delta}$ is strictly convex with respect to $\overline{\sg}$. 

Combining this discussion with Theorem \ref{COMN2} we deduce the following Proposition that was proved for the first time in \cite{0264-9381-32-19-195001} and that will be used fundamentally in the next section.

\begin{Proposition}\label{PIV}
Let $(\Sigma;g,N)$ be a static black hole data set. Then, for every $0<\epsilon<-1+\sqrt{2}$ there is $0<\delta<\delta_{0}$ such that $(\Sigma_{\delta}; \overline{g}=N^{-2\epsilon}g)$ is metrically complete and $\partial \Sigma_{\delta}$ is strictly convex (with respect to $\overline{g}$ and with respect to the inward normal).
\end{Proposition} 

The Riemannian spaces $(\Sigma_{\delta};\overline{g})$ have a metric, as discussed earlier, that we will denote by $\dist_{\overline{g}}^{\delta}$. 
The strict convexity of the boundaries as well as the metric completeness of  the spaces $(\Sigma_{\delta}; \overline{g})$ imply two basic, albeit important, geometric facts: 
\begin{enumerate}
\item[(i)] The distance $\dist_{\overline{g}}^{\delta}(p,q)$ between two points in $\Sigma_{\delta}$ is always realised by the length of a geodesic segment joining $p$ to $q$, and disjoint from $\partial \Sigma_{\delta}$ except, possibly, at the end-points $p$ and $q$.

\item[(ii)] Given a curve $I$ embedded in $\Sigma_{\delta}$ and with end-points $p$ and $q$, there is always a geodesic segment minimising length in the class of curves embedded in $\Sigma_{\delta}$, isotopic to $I$ and having the same end-points. The minimising segment is disjoint from $\partial \Sigma_{\delta}$ except, possibly, at the end points $p$ and $q$.

\end{enumerate} 

These properties allow us to make analysis as if the manifold $\Sigma_{\delta}$ were in practice boundary-less, and thus to import a series of results from {\it comparison geometry}, as developed for instance in \cite{MR2577473}, without worrying about the existence of the boundary.
\vs

\subsubsection{The structure of infinity}\label{STR}

The following proposition shows that static black hole data sets have only one end and moreover admit simple end cuts. 

\begin{Proposition}\label{KUNO} Let $(\sM;\sg,N)$ be a static black hole data set. Then $\sM$ has only one end. Moreover $(\Sigma;\sg)$ admits a simple end cut.
\end{Proposition}

\begin{proof} We work with the manifolds $(\Sigma_{\delta},\overline{g}=N^{-2\epsilon}g)$ from Proposition \ref{PIV}, with $0<\epsilon<-1+\sqrt{2}$ and $\delta=\delta(\epsilon)\leq \delta_{0}$. We argue first in a fixed $(\Sigma_{\delta}; \overline{g})$ and then let $\epsilon\rightarrow 0$. If $i_{\sM}>1$, i.e. if $\Sigma$ has at least two ends, then $\Sigma_{\delta}$ has also at least two ends. Hence $\Sigma_{\delta}$, (which has convex boundary) contains a line diverging through two of them. The presence of a line is relevant because,  even having $\partial\Sigma_{\delta} \neq \emptyset$, the geometry of $(\Sigma_{\delta}; \overline{g},N)$ is such (recall the discussion in Section \ref{CDPL}) that the {\it Splitting Theorem} as proved in \cite{MR2577473} applies 
\footnote{Theorem 6.1 in \cite{MR2577473} is stated for spaces with $Ric^{0}_{f}\geq 0$ and $f$ bounded. The boundedness of $f$ is required to have a Laplacian comparison for distance functions ($\S$ \cite{MR2577473} Theorem 1.1). No such condition on $f$ (hence on $N$, because $f=-(1+\epsilon)\ln N$) is required in our case, as we have $\overline{Ric}^{0}_{f}=\alpha\nabla f\nabla f$ with $\alpha>0$ and a Laplacian comparison holds without further assumptions ($\S$ \cite{MR2577473}, Theorem A.1).}.
More precisely, repeating line by line the proof of Theorem 6.1 in \cite{MR2577473}, one concludes that (see comments below after \ref{aWW}, \ref{bWW} and \ref{cWW}),  
\begin{enumerate}[labelindent=\parindent, leftmargin=*, label={\rm (\alph*)}, widest=a, align=left]
\item\label{aWW} there is a smooth Busemann function $b^{+}_{\epsilon}$, ($b^{+}$ in the notation of \cite{MR2577473}), with $|\nabla b^{+}_{\epsilon}|_{\overline{g}}=1$ and whose level sets are totally geodesic,
\item\label{bWW} the Ricci tensor is zero in the normal direction to the level sets, that is 
\be
\overline{Ric}(\nabla b^{+}_{\epsilon}, - )=0,
\ee
\item\label{cWW} $N$ is constant in the normal directions to the level sets, that is $\langle \nabla b_{\epsilon}^{+}, \nabla N\rangle_{\overline{g}}=0$.  
\end{enumerate}
The item \ref{aWW} is what is proved in Theorem 6.1 of \cite{MR2577473} and requires no comment. The items \ref{bWW} and \ref{cWW} follow instead from formula (6.11) in \cite{MR2577473} after recalling that in our case we have $\overline{Ric}^{0}_{f}=\alpha\nabla f\nabla f$, with $f=-(1+\epsilon)\ln N$ and $\alpha>0$. 

Of course \ref{aWW} implies that $\overline{g}$ locally splits. Namely, defining a coordinate $x$ by $x=b^{+}$, one can locally write $\overline{g}=dx^{2}+\overline{h}$, where $\overline{h}$ is the metric inherited from $\overline{g}$ on the level sets of $x$, that (under a natural identification) does not depend on $x$. 

The conclusions \ref{aWW},  \ref{bWW} and \ref{cWW} imply a contradiction as follows. Fix a point $p$ in $\Sigma_{\delta_{0}}^{\circ}$ and take a sequence $\epsilon_{i}\rightarrow 0$. Then, in a small but fixed neighbourhood $\mathcal{U}$ of $p$, the sequence $b^{+}_{\epsilon_{i}}$ sub-converges to a limit function $b^{+}_{0}$, with the same properties \ref{aWW}, \ref{bWW}, \ref{cWW} as each $b^{+}_{\epsilon_{i}}$ but now on $(\mathcal{U}; g,N)$, \footnote{The existence of the limit is easy to see because $|\nabla b^{+}_{\epsilon}|_{\overline{g}}=1$ and the level sets of $b^{+}_{\epsilon}$ are totally geodesic, (for every $\epsilon$). At every point the level set is just defined by geodesics perpendicular to $\nabla b^{+}_{\epsilon}$}. Hence $(\mathcal{U}; g)$ also splits. We claim that the Gaussian curvature $\gcur$ of the level sets of $b^{+}_{0}$ in $\mathcal{U}$ is zero. Indeed, as: (i) the level sets of $b_{0}^{+}$ are totally geodesic by \ref{aWW}, (ii) $Ric(\nabla b^{+}_{0},\nabla b^{+}_{0})=0$ by \ref{bWW}, and (iii) the scalar curvature $R$ of $g$ is zero by the static equations, then the Gauss-Codazzy equations yield $\gcur=0$. As $(\mathcal{U}; g)$ is flat then the static solution is flat everywhere by analyticity. The only flat static black hole data set with compact boundary is the Boost. As Boosts have only one end we reach a contradiction.  Hence $i_{\Sigma}=1$.

Let us prove now that $(\Sigma;g)$ admits simple cuts. Let $\{\mathcal{S}_{jk},j=0,1,2,\ldots,k=1,\ldots,k_{j}\}$ be an end cut. Suppose that $k_{j}>1$ for some $j\geq 0$. If we cut $\sM$ along $\mathcal{S}_{j1}$ we obtain a connected manifold, say $\sM'$, with two new boundary components, say $\mathcal{S}'_{1}$ and $\mathcal{S}'_{2}$, both of which are copies of $\mathcal{S}_{j1}$ (if cutting $\sM$ along $\mathcal{S}_{j1}$ results in two connected components then $k_{j}=1$ because of how simple cuts are constructed). Consider another copy of $\sM'$, denoted by $\sM''$ and denote the corresponding new boundary components as $\mathcal{S}''_{1}$ and $\mathcal{S}''_{2}$. By gluing $\mathcal{S}_{1}'$ to $\mathcal{S}''_{2}$ and $\mathcal{S}'_{2}$ to $\mathcal{S}''_{1}$ we obtain a static solution (a double cover of the original) with two ends, and one can proceed as earlier to obtain a contradiction. 
\end{proof}

\subsubsection{Horizons's types and properties}\label{HTT}

The following Proposition, about the structure of horizons, uses the completeness at infinity of $\hg$ and a pair of results due to Galloway \cite{4b6cb19bc94d4cf485e58571e3062f77}, \cite{MR1201655}.

\begin{Proposition}\label{SOFOR} Let $(\sM; g, N)$ be a static black hole data set. Then, either
\begin{enumerate}[labelindent=\parindent, leftmargin=*, label={\rm (\roman*)}, widest=a, align=left]
\item $(\sM; g, N)$ is a Boost and therefore $\partial \sM$ is a totally geodesic flat torus, or,
\item every component of $\partial \sM$ is a totally geodesic, weakly outermost, minimal sphere.
\end{enumerate}
\end{Proposition} 

\begin{proof} The idea is to prove that every component $H$ of $\partial \Sigma$ is a weakly outermost. Then, it is direct from Theorem 1.1 and 1.2 in \cite{4b6cb19bc94d4cf485e58571e3062f77} that either $H$ is a sphere or is a torus and if it is a torus then the whole space is a Boost. So let us prove that every component is weakly outermost.  

Let $\{H_{1},\ldots,H_{h}\}$, $h\geq 1$, be the set of horizons, i.e. the connected components of $\partial \Sigma$. Assume that there is an embedded orientable surface $\mathcal{S}$, homologous to one of the $H$'s, (say $H_{1}$), and with outer-mean curvature $\theta_{\mathcal{S}}$ strictly negative. For reference below define the negative constant $c$ as
\be
c=\sup\bigg\{\frac{\theta_{\mathcal{S}}(q)}{N(q)}: q\in \mathcal{S}\bigg\}
\ee 

Let $\{\mathcal{S}_{j},j=j_{0},j_{1},\ldots\}$ be a simple end cut of $(\Sigma;g)$ (Proposition \ref{KUNO}). For each $j$, let $\Omega(\partial \Sigma,\mathcal{S}_{j})$ be the closure of the connected component of $\Sigma\setminus \mathcal{S}_{j}$ containing $\partial \Sigma$. Let $\mathcal{U}$ be the closed region enclosed by $H_{1}$ and $\mathcal{S}$ and assume that $j_{0}$ is large enough that $\mathcal{S}_{j}\cap \mathcal{U}=\emptyset$ for all $j\geq j_{0}$. For every $j\geq j_{0}$ let $\mathcal{M}_{j}$ be the closed region enclosed by $\mathcal{S},H_{2},\ldots,H_{h}$ and $\mathcal{S}_{j}$, that is $\mathcal{M}_{j}=\Omega(\partial \Sigma, \mathcal{S}_{j})\setminus \mathcal{U}^{\circ}$. Finally let 
\be
\hat{\mathcal{M}}_{j}=\mathcal{M}_{j}\setminus (H_{2}\cup\ldots\cup H_{h})
\ee
and note that now $\partial \hat{\mathcal{M}}_{j}=\mathcal{S}\cup \mathcal{S}_{j}$. On $\hat{\mathcal{M}}_{j}$ consider the optical metric $\overline{g}=N^{-2}g$. The Riemannian space $(\hat{\mathcal{M}}_{j};\overline{g})$ is metrically complete, (roughly speaking the horizons $H_{i},i\geq 2$ have been blown to infinity). 

Now, for every $j\geq j_{0}$ let $\gamma_{j}$ be the $\overline{g}$-geodesic segment inside $\hat{\mathcal{M}}_{j}$, realising the $\overline{g}$-distance between $\mathcal{S}$ and $\mathcal{S}_{j}$. 
The segments $\gamma_{j}$ are perpendicular to $\mathcal{S}$. Also, as they are length-minimising the $\overline{g}$-expansion $\overline{\theta}$ of the congruence of $\overline{g}$-geodesics emanating perpendicularly from $\mathcal{S}$, remains finite all along $\gamma_{j}$. Let $s\in [0,s_{j}]$ be the $g$-arc-length of $\gamma_{j}$ measured from $\mathcal{S}$. Note that $s$ is not the arc-length with respect to $\overline{g}$, that would be natural. We are going to use this parameterisation of $\gamma_{j}$ below. Observe that $s_{j}\rightarrow \infty$ as $j\rightarrow \infty$. 

Along $\gamma_{j}(s)$ let 
\be\label{RRR}
F(s)=\overline{\theta}(\gamma_{j}(s))+\frac{2}{N^{2}(\gamma_{j}(s))}\frac{d N(\gamma_{j}(s))}{ds}
\ee
Then, as shown by Galloway \cite{MR1201655} (see also \cite{MR3077927}), the function $F$ satisfies the following differential inequality
\be
\frac{dF}{ds}\leq -\frac{N}{2}F^{2}
\ee
Now, a simple computation shows that $F(0)=\theta(0)/N(0)\leq c<0$. But from (\ref{RRR}) it is easily deduced that if 
\be
\int_{0}^{s_{j}} N(\gamma_{j}(s))ds>-\frac{2}{c}
\ee
then there is $s^{*}\in (0,s_{j})$ such that $F(s^{*})=-\infty$, thus $\overline{\theta}(s^{*})=-\infty$ and the $\gamma_{j}$ would not be $\overline{g}$-length minimising. Thus, a contradiction is reached if we prove that $\int_{0}^{s_{j}} N(\gamma_{j}(s))ds\rightarrow \infty$. But his follows from the completeness of the metric $\hg=N^{2}g$ from Theorem \ref{COMN2}.

\end{proof}

\subsubsection{The asymptotic of isolated systems.}\label{TAIS}

Theorem \ref{COMN2} shows that if $N>0$ and $\partial \Sigma$ is compact then $(\Sigma; \hg=N^{2}g)$ is metrically complete. On the other hand it was proved in \cite{MR3233266}, \cite{MR3233267}, that if $\Sigma$ is diffeomorphic to $\mathbb{R}^{3}$ minus a ball and $\hg$ is complete then the space $(\Sigma;g,N)$ is asymptotically flat. Combining these two results we obtain that: if $\Sigma$ minus a compact set $K$ is diffeomorphic to $\mathbb{R}^{3}$ minus a closed ball then the data set $(\Sigma; g,N)$ is asymptotically flat. Asymptotic flatness is thus characterised only by the asymptotic topology of $\Sigma$. 

This fact has physically interesting consequences. Following physical intuition define a {\it static isolated system} as a static space-time $(\mathbb{R}\times\Sigma; -N^{2}dt^{2}+g)$, ($\partial \Sigma=\emptyset$ and $(\Sigma; g)$ metrically complete), for which there is a set $K\subset \Sigma$ such that $\Sigma\setminus K$ is diffeomorphic to $\mathbb{R}^{3}$ minus a closed ball and such that the region $\mathbb{R}\times (\Sigma\setminus K)$ is vacuum (i.e. matter lies only in $\mathbb{R}\times K$). The most obvious example of static isolated system one can think of is that of body like a planet or a star. Then, using what we explained in the previous paragraph, static isolated systems are always asymptotically flat. This conclusion was reached in \cite{0264-9381-32-19-195001} but requiring as part of the definition of static isolated system that the space-time is null geodesically complete at infinity. What we are showing here is that this condition is indeed unnecessary and the completeness of the hypersurface $(\Sigma; g)$ is sufficient.   

\section{Global properties of the lapse}\label{GLOBP}

We aim to prove that the lapse $N$ of any black hole data set is bounded away from zero at infinity, namely that there is $c>0$ such that for any divergent sequence $p_{n}$ we have $\lim N(p_{n})\geq c$.
\begin{Theorem}\label{BNFB} Let $(\sM;g,N)$ be a static black hole data set. Then, $N$ is bounded away from zero at infinity.
\end{Theorem}
The proof of this theorem will follow after some propositions that we state and prove below.
\begin{Proposition}\label{PT1} Let $(\Sigma_{\delta};\overline{g})$ be a space as in Proposition \ref{PIV}, with $0<\epsilon<1/4$. Let $p$ and $q$ be two different points in $\Sigma_{\delta}$ and let $\gamma:[0,L]\rightarrow \sM_{\delta}$ be a $\overline{g}$-geodesic (parameterised with the arc-length $\overline{s}$) starting at $p$ and ending at $q$ and minimising the $\overline{g}$-length in its own isotopy class (with fixed end points). Then, for any $0<s<t<L$ we have
\be\label{BFN}
-\sqrt{50\bigg[\frac{(t-s)}{s}+\frac{(t-s)}{L-t}\bigg]}\leq \ln \bigg[\frac{N(\gamma(t))}{N(\gamma(s))}\bigg]\leq \sqrt{50\bigg[\frac{(t-s)}{s}+\frac{(t-s)}{L-t}\bigg]}
\ee 
\end{Proposition}
Note that in this statement, $s$, $t-s$ and $L-t$ are, respectively, the $\overline{\sg}$-distances along $\gamma$ between the pairs of points $(p,\gamma(s))$, $(\gamma(s),\gamma(t))$ and $(\gamma(t),q)$.

\begin{proof} Let $f$ and $\alpha$ be as in Proposition \ref{FELIZ}. Let $\gamma$, $s$ and $t$ be as in the hypothesis. Let $\theta(\overline{s})$ be the expansion along $\gamma$ of the congruence of geodesics emanating from $p$, where $\overline{s}$ is the arc-length. From (\ref{OOO}) we can write 
\be\label{RBE2}
\overline{Ric}^{\alpha/2}_{f}=\overline{Ric}+\overline{\nabla}\overline{\nabla}f-\frac{\alpha}{2}\overline{\nabla}f\overline{\nabla}f=\frac{\alpha}{2}\overline{\nabla}f\overline{\nabla}f
\ee
where $0<\alpha$ because $0<\epsilon<1/4<-1+\sqrt{2}$. Let $\theta_{f}=\theta-f'$ where $f'=df(\gamma(\overline{s}))/d\overline{s}$. As shown in \cite{MR2577473}, (\ref{RBE2}) implies that,
\be
\theta_{f}'\leq -\frac{1}{2/\alpha+3}\theta_{f}^{2}-\frac{\alpha}{2}(f')^{2}=-a^{2}\theta_{f}^{2}-b^{2}\bigg(\frac{N'}{N}\bigg)^{2}
\ee
where $'=d/d\overline{s}$ and
\be\label{514}
a^{2}=\frac{1}{2/\alpha+\epsilon},\quad \text{and}\quad b^{2}=\frac{(1+\epsilon)^{2}\alpha}{2}
\ee

From the differential inequality $\theta_{f}'\leq -a^{2}\theta_{f}^{2}$ we deduce,
\be\label{INEC1}
\theta_{f}(s)\leq \frac{1}{a^{2}s}
\ee
and also we deduce 
\be\label{INEC2}
\theta_{f}(t)\geq -\frac{1}{a^{2}(L-t)}
\ee
because if $\theta_{f}(s)<-\frac{1}{L-t}$ then there exists $r$, with $t<r<L$, for which $\theta_{f}(r)=-\infty$, and therefore $\theta(r)=-\infty$, contradicting that $\gamma$ is length minimising within its isotopy class.

Hence, we can use (\ref{INEC1}) and (\ref{INEC2}) and $\theta_{f}'\leq -b^{2}(N'/N)^{2}$ to deduce
\begin{align}
\bigg|\ln \frac{N(t)}{N(s)}\bigg|^{2}&=\bigg|\int_{s}^{t}\frac{N'}{N}d\overline{s}\bigg|^{2}\leq (t-s)\int_{s}^{t}\bigg(\frac{N'}{N}\bigg)^{2}d\overline{s}\\
&\leq (t-s)\frac{1}{b^{2}}(\theta_{f}(s)-\theta_{f}(t))\leq \frac{(t-s)}{a^{2}b^{2}}\bigg(\frac{1}{s}+\frac{1}{L-t}\bigg) 
\end{align}
which gives (\ref{BFN}) if one observes that $1/a^{2}b^{2}\leq 50$, after a short computation involving (\ref{514}), the form of $\alpha$ from Proposition \ref{FELIZ}, and the fact that $\epsilon<1/4$.
\end{proof}

\begin{Proposition}\label{PT2} Let $(\Sigma;g,N)$ be a static black hole data set. Let $\mathcal{S}_{1}$ and $\mathcal{S}_{2}$ be two disjoint, connected, compact, boundary-less and orientable surfaces, embedded in $\sM^{\circ}$. Let $W:\mathbb{R}\rightarrow \sM^{\circ}$ be a smooth embedding, intersecting $\mathcal{S}_{1}$ and $\mathcal{S}_{2}$ only once and transversely and with $W(t)$ diverging as $t\rightarrow \pm\infty$. Then, there is $p_{1}\in \mathcal{S}_{1}$ and $p_{2}\in \mathcal{S}_{2}$ such that $N(p_{1})=N(p_{2})$.
\end{Proposition}

\begin{proof} We work in a manifold $(\sM_{\delta};\overline{\sg})$ as in Proposition \ref{PIV} and with $0<\epsilon<1/4$. Assume thus that $\delta$ is small enough that $(W\cup \mathcal{S}_{1}\cup \mathcal{S}_{2})\subset \sM^{\circ}_{\delta}$. Orient $W$ in the direction of increasing $t$. Orient also $\mathcal{S}_{1}$ and $\mathcal{S}_{2}$ in such a way that the intersection number between $\mathcal{S}$ and $W$, and between $\mathcal{S}_{2}$ and $W$, are both equal to one. All intersection numbers below are defined with respect to these orientations.   

Redefine the parameter $t$ if necessary to have $W(-1)\in \mathcal{S}_{1}$ and $W(1)\in \mathcal{S}_{2}$. Then, for every natural number $m\geq 1$ let $\gamma_{m}(\overline{s})$ be a $\overline{g}$-geodesic minimising the $\overline{g}$-length among all the curves embedded in $\sM_{\delta}^{\circ}$, with end points $W(-1-m)$ and $W(1+m)$ and having non-zero intersection number with $\mathcal{S}_{1}$ and $\mathcal{S}_{2}$,
\footnote{The existence of such geodesic is as follows. Let ${\mathcal{C}}$ be the family of all curves joining $W(-1-m)$ and $W(1+m)$ and having non-zero intersection number with $\mathcal{S}_{1}$ and $\mathcal{S}_{2}$. As the intersection number is an isotopy-invariant, the family ${\mathcal{C}}$ is a union of isotopy classes. In each class consider a representative minimising length inside the class (recall the discussion in Section \ref{CDPL}). Let $C_{i}$ be a sequence of such representatives and (asymptotically) minimising length in the family ${\mathcal{C}}$. Such sequence has a convergent subsequence, to, say, $C_{\infty}$. As for $i\geq i_{0}$ with $i_{0}$ big enough, $C_{i}$ is isotopic to $C_{\infty}$ we conclude that $C_{\infty}\in {\mathcal{C}}$ as wished.}.
We denoted by $\overline{s}$ the $\overline{\sg}$-arc length starting from $W(-1-m)$. The $\overline{\sg}$-length of $\gamma_{m}$ is denoted by $L_{m}$. 

We want to prove that there are points $p^{1}_{m}:=\gamma_{m}(\overline{s}^{1}_{m})\in \mathcal{S}_{1}$ and $p^{2}_{m}:=\gamma_{m}(\overline{s}^{2}_{m})\in \mathcal{S}_{2}$, (for some $\overline{s}^{1}_{m}$ and $\overline{s}^{2}_{m}$), with $|\overline{s}^{2}_{m}-\overline{s}^{1}_{m}|$ uniformly bounded above. Once this is done the proof is finished as follows. As the initial and final points $W(-1-m)$ and $W(1+m)$ get further and further away from $\mathcal{S}_{1}$ and $\mathcal{S}_{2}$, then we have $\overline{s}^{1}_{m}\rightarrow \infty$, $\overline{s}^{2}_{m}\rightarrow \infty$, $L_{m}-\overline{s}^{2}_{m}\rightarrow \infty$, and $L_{m}-\overline{s}^{1}_{m}\rightarrow \infty$. Therefore we can rely in Proposition \ref{PT1} used with $\gamma=\gamma_{m}$, $\gamma(s)=p^{1}_{m}$, and $\gamma(t)=p^{2}_{m}$, to conclude that 
\be
\lim_{m\rightarrow \infty} |N(p^{1}_{m})-N(p^{2}_{m})|= 0
\ee  
Hence, if $p_{1}$ is an accumulation point of $\{p^{1}_{m}\}$ and $p_{2}$ an accumulation point of $\{p^{2}_{m}\}$ we will have $N(p_{1})=N(p_{2})$ as desired.

Consider now the set of embedded curves $X:[-1,1]\rightarrow \sM^{\circ}$, starting at $\mathcal{S}_{1}$ and transversely to it, ending at $\mathcal{S}_{2}$ and transversely to it, and not intersecting $\mathcal{S}_{1}$ and $\mathcal{S}_{2}$ except of course at the initial and final points. There are at most four classes of curves $X$, distinguished according to the direction to which the vectors $X'(-1)$ and $X'(1)$ point. For each non-empty class fix a representative, so there are at most four of them, and let $B$ be a common upper bound of their lengths.  

Without loss of generality assume that each $\gamma_{m}$, as defined earlier, intersects $\mathcal{S}_{1}$ and $\mathcal{S}_{2}$ transversely\footnote{Otherwise use suitable small deformations}. Let also $\{\gamma_{m}(\overline{s}^{1}_{1m}),\ldots,\gamma_{m}(\overline{s}^{1}_{l_{1}m})\}$ and $\{\gamma_{m}(\overline{s}^{2}_{1m}),\ldots,\gamma_{m}(\overline{s}^{2}_{l_{2}m})\}$ be the points of intersection of $\gamma_{m}$ with $\mathcal{S}_{1}$ and $\mathcal{S}_{2}$ respectively. For each $m$ choose any two $\overline{s}^{1}_{i_{1}m}$ and $\overline{s}^{2}_{i_{2}m}$ consecutive, namely that the open interval 
\be
(\min\{\overline{s}^{1}_{i_{1}m},\overline{s}^{2}_{i_{2}m}\},\max\{\overline{s}^{1}_{i_{1}m},\overline{s}^{2}_{i_{2}m}\})
\ee
does not contain any of the elements $\{\overline{s}^{1}_{1m},\ldots,\overline{s}^{1}_{l_{1}m};\overline{s}^{2}_{1m},\ldots,\overline{s}^{2}_{l_{2}m}\}$. Without loss of generality we assume that $\overline{s}^{1}_{i_{1}m}<\overline{s}^{2}_{i_{2}m}$ for all $m$. 

To simplify notation let $\overline{s}^{1}_{m}:=\overline{s}^{1}_{i_{1}m}$ and $\overline{s}^{2}_{m}:=\overline{s}^{2}_{i_{2}m}$. The curves $X_{m}(\overline{s}):=\gamma_{m}(\overline{s})$, $\overline{s}\in [\overline{s}^{1}_{m},\overline{s}^{2}_{m}]$, can be thought (after reparameterisation) as belonging to one of the four classes of curves $X$ described above. For every $m$ let then $\hat{X}_{m}$ be the representative, chosen earlier, of the class to which $X_{m}$ belongs. 

We compare now the length of $\gamma_{m}$ with the length of a competitor curve, that we denote by $\hat{\gamma}_{m}$, and that is constructed out of $\hat{X}_{m}$ and $\gamma_{m}$ itself. The construction of $\hat{\gamma}_{m}$ is better described in words. Starting from $\gamma_{m}(0)$ we move forward through $\gamma_{m}$, reach $\mathcal{S}_{1}$ at $\gamma_{m}(\overline{s}^{1}_{m})$, and cross it slightly. From there we move through a curve very close to $\mathcal{S}_{1}$ and of length less than $2\diam(\mathcal{S}_{1})$ until reaching a point in $\hat{X}_{m}$. Then we move through $\hat{X}_{m}$ until a point right before $\mathcal{S}_{2}$. Finally we move through a curve very close to $\mathcal{S}_{2}$ and of length less than $2\diam(\mathcal{S}_{2})$ until reaching a point in $\gamma_{m}$ right before $\gamma_{m}(\overline{s}^{2}_{m})$, from which we move through $\gamma_{m}$ until reaching $\gamma_{m}(L_{m})$. Clearly $\gamma_{m}$ has the same intersection numbers with $\mathcal{S}_{1}$ and $\mathcal{S}_{2}$ as $\gamma_{m}$ has, hence non-zero. Thus, by the definition of $\gamma_{m}$ we have,
\be
\length(\gamma_{m})\leq \length(\hat{\gamma}_{m})
\ee
But we have 
\be
\length(\gamma_{m})=\overline{s}^{1}_{m}+(\overline{s}^{2}_{m}-\overline{s}^{1}_{m})+(L_{m}-\overline{s}^{2}_{m})
\ee
and (if the construction of $\hat{\gamma}_{m}$ is fine enough)
\be
\length(\hat{\gamma}_{m})\leq \overline{s}^{1}_{m}+2\diam(\mathcal{S}_{1})+\length(\hat{X}_{m})+2\diam(\mathcal{S}_{2})+(L_{m}-\overline{s}^{2}_{m})
\ee
Hence, as $\length(\hat{X}_{m})\leq B$ we conclude that 
\be
\overline{s}^{2}_{m}-\overline{s}^{1}_{m}\leq B+2\diam(\mathcal{S}_{1})+2\diam(\mathcal{S}_{2})
\ee
That is, $|\overline{s}^{2}_{m}-\overline{s}^{1}_{m}|$ is uniformly bounded as wished.
\end{proof}

Let us introduce the setup required for the next Proposition \ref{PT3} and for the proof of Theorem \ref{BNFB}. Although it was proved earlier that static black hole data sets have only one end, below we will work as if the manifold could have more than one end. The reason for this is that the framework below is valid in higher dimensions, and this could help to investigate whether higher dimensional vacuum static black holes have also only one end. The proof of this fact that we gave earlier holds only in dimension three.

Choose $\sM_{i},i=1,\ldots, i_{\sM}\geq 1$ a set of non-compact and {\it connected} regions of $\sM^{\circ}$, with compact (and smooth) boundaries, each containing only one end, and the union covering $\sM$ except for a {\it connected} set of compact closure, (i.e. $\sM\setminus (\cup \sM_{i}^{\circ})$ is compact and connected). For each end $\sM_{i}$ we consider an end cut $\{\mathcal{S}_{ijk}, j\geq 0, k=1,\ldots,k_{ij}\}$. 

The surfaces $\mathcal{S}_{ijk}$ are considered only to serve as a `reference'. Their geometry plays no role. The condition that the union of the ends $\Sigma_{i}$ covers $\Sigma$ except for a connected set of compact closure will be technically relevant in the proof below. It ensures that given any two $\mathcal{S}_{ijk}$ and $\mathcal{S}_{i'j'k'}$ with either: $i\neq i'$ ($j,k,j',k'$ any), or $i=i'$, $j=j'$ ($k,k'$ any), one can always find an immersed curve $W:\mathbb{R}\rightarrow \Sigma$ intersecting $\mathcal{S}_{ijk}$ and $\mathcal{S}_{i'j'k'}$ only once and such that $W(t)$ diverges as $t\rightarrow\pm \infty$. This fact follows directly from the definition of end cut.

\begin{Proposition}\label{PT3} (setup above) Let $(\sM;g,N)$ be a static black hole data set. Then,
\begin{enumerate}
\item If $i_{\sM}>1$, then for any $\mathcal{S}_{ijk}$ and $\mathcal{S}_{i'j'k'}$, with $i\neq i'$, there are points $p\in \mathcal{S}_{ijk}$ and $p'\in \mathcal{S}_{i'j'k'}$ such that $N(p)=N(p')$.
\item If $i_{\sM}=1$, then for every $j$ with $k_{1j}>1$ and $1\leq k\neq k'\leq k_{1j}$, there are points $p\in \mathcal{S}_{1jk}$ and $p'\in \mathcal{S}_{1jk'}$ such that $N(p)=N(p')$. 
\end{enumerate}
\end{Proposition}

\begin{proof} If $i_{\sM}>1$ then we can easily construct an embedding $W:\mathbb{R}\rightarrow \sM^{\circ}$ intersecting the manifolds $\mathcal{S}_{ijk}$ and $\mathcal{S}_{i'j'k'}$ only once and with $W(t)\rightarrow \infty$ as $t\rightarrow \pm \infty$. The existence of $p\in \mathcal{S}_{ijk}$ and $p'\in \mathcal{S}_{i'j'k'}$ for which $N(p)=N(p')$ then follows from Proposition \ref{PT2}. The case $i_{\sM}=1$ is treated in exactly the same way.
\end{proof}

We are ready to prove Theorem \ref{BNFB}.

\begin{proof}[\it Proof of Theorem \ref{BNFB}.] We use the same setup as in Proposition \ref{PT3}. Also we let $\mathcal{S}_{ij}:=\cup_{k=1}^{k=k_{ij}}\mathcal{S}_{ijk}$ and given $j'>j$, $\mathcal{U}_{i;jj'}$ denotes the closed region enclosed by $\mathcal{S}_{ij}$ and $\mathcal{S}_{ij'}$. Also, given a closed set $C$, we let $\min\{N;C\}:=\min\{N(x):x\in C\}$ and similarly for $\max\{N;C\}$.

We want to show that $N$ is bounded from below away from zero at every one of the ends $\sM_{i}$. We distinguish two cases: $i_{\Sigma}>1$ and $i_{\Sigma}=1$. 

{\it Case $i_{\sM}>1$}. Without loss of generality we prove this only for $\sM_{1}$. Let us fix a surface $\mathcal{S}_{2j_{0}k_{0}}$ in $\sM_{2}$. By Proposition \ref{PT3} we know that at every $\mathcal{S}_{1jk}$ we have
\be\label{BEFF}
0<\min\{N;\mathcal{S}_{2j_{0}k_{0}}\}\leq \max\{N;\mathcal{S}_{1jk}\}
\ee
On the other hand the Harnak estimate (\ref{EQHARN1}) in Proposition \ref{MAXMINU11} gives us 
\be\label{GO}
\max\{N;\mathcal{S}_{1jk}\}\leq \eta' \min\{N;\mathcal{S}_{1jk}\}
\ee
where $\eta'$ is independent of $j$ and $k$. Combined with (\ref{BEFF}) this gives us the bound
\be\label{FDDA}
0<\eta''<\min\{N;\mathcal{S}_{1jk}\}
\ee
where $\eta''$ is independent of $j$ and $k$. Now, recall that the manifolds $\mathcal{U}_{1;j,j+1},j=0,1,\ldots$ cover $\sM_{1}$ up to a set of compact closure and that for each $j$, $\partial \mathcal{U}_{1;j,j+1}$ is the union of the surfaces $\mathcal{S}_{1jk};k=1,\ldots,k_{1j}$ and $\mathcal{S}_{1,j+1,k};k=1,\ldots,k_{1,j+1}$. Therefore by (\ref{FDDA}) and the maximum principle we deduce,
\be
0< \eta''< \min\{N;\partial \mathcal{U}_{1;j,j+1}\}\leq \min\{N; \mathcal{U}_{1;j,j+1}\}
\ee
from which the lower bound for $N$ away from zero over $\sM_{1}$ follows.

{\it Case $i_{\sM}=1$}. We observe first that, as in this case $\sM_{1}$ is the only end and as $N=0$ on $\partial \sM$, then $N$ cannot go uniformly to zero at infinity (this would violate the maximum principle).  We prove now that, if there is a diverging sequence $p_{l}$ such that $N(p_{l})\rightarrow 0$, then $N$ must go to zero uniformly at infinity. The proof will then be finished.

As $i_{\Sigma}=1$ we will remove the index $i=1$ everywhere from now on. For every $l$ let $j_{l}$ be such that $p_{l}\in \mathcal{U}_{j_{l},j_{l}+1}$ and let $\mathcal{U}^{c}_{j_{l},j_{l}+1}$ be the connected component of $\mathcal{U}_{j_{l},j_{l}+1}$ containing $p_{l}$. By the maximum principle we have
\be
\min\{N;\partial \mathcal{U}^{c}_{j_{l},j_{l}+1}\}\leq \min\{N;\mathcal{U}^{c}_{j_{l},j_{l}+1}\}\leq N(p_{l})
\ee
Therefore we can extract a sequence of connected components of $\partial \mathcal{U}^{c}_{j_{l},j_{l}+1}$, denoted by $\mathcal{S}_{j^{l}k_{l}}$ ($j^{l}$ is either $j_{l}$ or $j_{l}+1$), such that 
\be
\min\{N; \mathcal{S}_{j^{l}k_{l}}\}\rightarrow 0
\ee
From this and (\ref{GO}) we obtain
\be\label{GFGF}
\max\{N;\mathcal{S}_{j^{l}k_{l}}\}\rightarrow 0
\ee
Then, by Proposition \ref{PT3} we have
\be\label{POPP}
\min\{N;\mathcal{S}_{j^{l}k}\}\leq \max\{N;\mathcal{S}_{j^{l}k_{l}}\}
\ee
(note the difference in the subindexes $k$ and $k_{l}$) for all $k= 1,\ldots,k_{j^{l}}$ (it could be of course $k_{j^{l}}=1$). Using (\ref{GO}) in the left hand side of (\ref{POPP}) and using (\ref{GFGF}) we get
\be\label{SSAA}
\max\{N;\mathcal{S}_{j^{l}}\}\rightarrow 0
\ee
By the maximum principle again we deduce for any $l'>l$ the inequality
\be
\max\{N;\mathcal{U}_{j^{l}j^{l'}}\}\leq \max\{\max\{N;\mathcal{S}_{j^{l}}\}; \max\{N;\mathcal{S}_{j^{l'}}\}\}
\ee
Taking the limit $l'\rightarrow \infty$ we deduce that
the supremum of $N$ over the unbounded connected component of $\sM\setminus \mathcal{S}_{j^{l}}$ is less or equal than the maximum of $N$ over $\mathcal{S}_{j^{l}}$. Hence $N$ must tend uniformly to zero at infinity because of (\ref{SSAA}).
\end{proof}

\bibliographystyle{plain}
\bibliography{Master}

\begin{thebibliography}{10}

\bibitem{2015arXiv150404563A}
V.~{Agostiniani} and L.~{Mazzieri}.
\newblock {On the geometry of the level sets of bounded static potentials}.
\newblock {\em ArXiv e-prints}, April 2015.

\bibitem{MR1726233}
M.~T. Anderson.
\newblock Scalar curvature, metric degenerations and the static vacuum
  {E}instein equations on {$3$}-manifolds. {I}.
\newblock {\em Geom. Funct. Anal.}, 9(5):855--967, 1999.

\bibitem{MR1837365}
M.~T. Anderson.
\newblock Scalar curvature, metric degenerations, and the static vacuum
  {E}instein equations on 3-manifolds. {II}.
\newblock {\em Geom. Funct. Anal.}, 11(2):273--381, 2001.

\bibitem{MR1452867}
Michael~T. Anderson.
\newblock Scalar curvature and geometrization conjectures for {$3$}-manifolds.
\newblock In {\em Comparison geometry ({B}erkeley, {CA}, 1993--94)}, volume~30
  of {\em Math. Sci. Res. Inst. Publ.}, pages 49--82. Cambridge Univ. Press,
  Cambridge, 1997.

\bibitem{MR1806984}
Michael~T. Anderson.
\newblock On stationary vacuum solutions to the {E}instein equations.
\newblock {\em Ann. Henri Poincar\'e}, 1(5):977--994, 2000.

\bibitem{MR1809792}
Michael~T. Anderson.
\newblock On the structure of solutions to the static vacuum {E}instein
  equations.
\newblock {\em Ann. Henri Poincar\'e}, 1(6):995--1042, 2000.

\bibitem{MR3064190}
Michael~T. Anderson and Marcus~A. Khuri.
\newblock On the {B}artnik extension problem for the static vacuum {E}instein
  equations.
\newblock {\em Classical Quantum Gravity}, 30(12):125005, 33, 2013.

\bibitem{MR996396}
Robert Bartnik.
\newblock New definition of quasilocal mass.
\newblock {\em Phys. Rev. Lett.}, 62(20):2346--2348, 1989.

\bibitem{MR876598}
Gary~L. Bunting and A.~K.~M. Masood-ul Alam.
\newblock Nonexistence of multiple black holes in asymptotically {E}uclidean
  static vacuum space-time.
\newblock {\em Gen. Relativity Gravitation}, 19(2):147--154, 1987.

\bibitem{Robinson}
Robinson.~D. C.
\newblock Four decades of black hole uniqueness.
\newblock In M.~Scott D.~L.~Wiltshire, M.~Visser, editor, {\em The Kerr
  spacetime: Rotating black holes in General Relativity}, pages 115--143.
  Cambridge University Press, 2009.

\bibitem{MR2741248}
Jeffrey~S. Case.
\newblock The nonexistence of quasi-{E}instein metrics.
\newblock {\em Pacific J. Math.}, 248(2):277--284, 2010.

\bibitem{MR1201655}
Gregory~J. Galloway.
\newblock On the topology of black holes.
\newblock {\em Comm. Math. Phys.}, 151(1):53--66, 1993.

\bibitem{4b6cb19bc94d4cf485e58571e3062f77}
{Gregory J.} Galloway.
\newblock Rigidity of marginally trapped surfaces and the topology of black
  holes.
\newblock {\em Communications in Analysis and Geometry}, 16(1):217--229, 1
  2008.

\bibitem{MR2238889}
Gregory~J. Galloway and Richard Schoen.
\newblock A generalization of {H}awking's black hole topology theorem to higher
  dimensions.
\newblock {\em Comm. Math. Phys.}, 266(2):571--576, 2006.

\bibitem{Hagen1973}
H.~M{\"u}ller~Zum Hagen, David~C. Robinson, and H.~J. Seifert.
\newblock Black holes in static vacuum space-times.
\newblock {\em General Relativity and Gravitation}, 4(1):53--78, 1973.

\bibitem{Israel}
Israel.
\newblock Event horizons in static vacuum space-times.
\newblock {\em Phys. Review}, vol.164:5:1776--1779, 1967.

\bibitem{MR3037574}
Jeffrey~L. Jauregui.
\newblock Fill-ins of nonnegative scalar curvature, static metrics, and
  quasi-local mass.
\newblock {\em Pacific J. Math.}, 261(2):417--444, 2013.

\bibitem{Jordan2009}
Pascual Jordan, J{\"u}rgen Ehlers, and Wolfgang Kundt.
\newblock Republication of: Exact solutions of the field equations of the
  general theory of relativity.
\newblock {\em General Relativity and Gravitation}, 41(9):2191--2280, 2009.

\bibitem{Schwarzschild}
Schwarzschild Karl.
\newblock Uber das gravitationsfeld eines massepunktes nach der einsteinschen
  theorie.
\newblock {\em Preussische Akademie der Wissenschaften, Sitzungberichte}, pages
  189--196, 1916.

\bibitem{MR1501301}
Edward Kasner.
\newblock An algebraic solution of the {E}instein equations.
\newblock {\em Trans. Amer. Math. Soc.}, 27(1):101--105, 1925.

\bibitem{MR1501305}
Edward Kasner.
\newblock Solutions of the {E}instein equations involving functions of only one
  variable.
\newblock {\em Trans. Amer. Math. Soc.}, 27(2):155--162, 1925.

\bibitem{KOROTKIN1994229}
D.~Korotkin and H.~Nicolai.
\newblock The ernst equation on a riemann surface.
\newblock {\em Nuclear Physics B}, 429(1):229 -- 254, 1994.

\bibitem{94aperiodic}
D.~Korotkin and H.~Nicolai.
\newblock A periodic analog of the schwarzschild solution.
  arxiv:gr-qc/9403029v1, 1994.

\bibitem{kunzle1971}
H.~P. Kunzle.
\newblock On the spherical symmetry of a static perfect fluid.
\newblock {\em Comm. Math. Phys.}, 20(2):85--100, 1971.

\bibitem{Levi-Civita2011b}
T.~Levi-Civita.
\newblock Republication of: Einsteinian ds2 in newtonian fields. ix: The analog
  of the logarithmic potential.
\newblock {\em General Relativity and Gravitation}, 43(8):2321--2330, 2011.

\bibitem{Levi-Civita2011a}
T.~Levi-Civita.
\newblock Republication of: The physical reality of some normal spaces of
  bianchi.
\newblock {\em General Relativity and Gravitation}, 43(8):2307--2320, 2011.

\bibitem{MR1216638}
Zhong-dong Liu.
\newblock Ball covering property and nonnegative {R}icci curvature outside a
  compact set.
\newblock In {\em Differential geometry: {R}iemannian geometry ({L}os
  {A}ngeles, {CA}, 1990)}, volume~54 of {\em Proc. Sympos. Pure Math.}, pages
  459--464. Amer. Math. Soc., Providence, RI, 1993.

\bibitem{MR3077927}
Marc Mars and Martin Reiris.
\newblock Global and uniqueness properties of stationary and static spacetimes
  with outer trapped surfaces.
\newblock {\em Comm. Math. Phys.}, 322(2):633--666, 2013.

\bibitem{0264-9381-32-19-195001}
Martin.
\newblock The asymptotic of static isolated systems and a generalized
  uniqueness for schwarzschild.
\newblock {\em Classical and Quantum Gravity}, 32(19):195001, 2015.

\bibitem{PhysRevD.35.455}
R.~C. Myers.
\newblock Higher-dimensional black holes in compactified space-times.
\newblock {\em Phys. Rev. D}, 35:455--466, Jan 1987.

\bibitem{MR2243772}
Peter Petersen.
\newblock {\em Riemannian geometry}, volume 171 of {\em Graduate Texts in
  Mathematics}.
\newblock Springer, New York, second edition, 2006.

\bibitem{MR2919527}
Martin Reiris.
\newblock Static solutions from the point of view of comparison geometry.
\newblock {\em J. Math. Phys.}, 53(1):012501, 31, 2012.

\bibitem{MR3233266}
Martin Reiris.
\newblock Stationary solutions and asymptotic flatness {I}.
\newblock {\em Classical Quantum Gravity}, 31(15):155012, 33, 2014.

\bibitem{MR3233267}
Martin Reiris.
\newblock Stationary solutions and asymptotic flatness {II}.
\newblock {\em Classical Quantum Gravity}, 31(15):155013, 18, 2014.

\bibitem{Reiris2017}
Mart{\'i}n Reiris.
\newblock On static solutions of the einstein-scalar field equations.
\newblock {\em General Relativity and Gravitation}, 49(3):46, 2017.

\bibitem{RobinsonII}
D.C. Robinson.
\newblock A simple proof of the generalization of israel's theorem.
\newblock {\em General Relativity and Gravitation}, 8(8):695--698, 1977.

\bibitem{MR2577473}
Guofang Wei and Will Wylie.
\newblock Comparison geometry for the {B}akry-{E}mery {R}icci tensor.
\newblock {\em J. Differential Geom.}, 83(2):377--405, 2009.

\bibitem{ANDP:ANDP19173591804}
Hermann Weyl.
\newblock Zur gravitationstheorie.
\newblock {\em Annalen der Physik}, 359(18):117--145, 1917.

\end{thebibliography}

\end{document}